\documentclass[pdflatex,sn-mathphys-num]{sn-jnl}

\usepackage{graphicx}
\usepackage{amsmath,amssymb,amsfonts}
\usepackage{amsthm}
\usepackage{mathrsfs}
\usepackage[title]{appendix}
\usepackage{xcolor}
\usepackage{booktabs}
\usepackage{subcaption}
\usepackage{mdframed}
\let\orcidlogo\relax
\usepackage{orcidlink}

\theoremstyle{thmstyleone}
\newtheorem{theorem}{Theorem}
\newtheorem{lemma}{Lemma}
\newtheorem{proposition}{Proposition}
\newtheorem{observation}{Observation}

\begin{document}

\title[An FPTAS for two-machine open-shop scheduling]{An FPTAS for Two-Machine Open-Shop Scheduling with a Single Unavailability Interval}

\author[1]{\fnm{Hao} \sur{Lu}}
\author*[2]{\fnm{Yuan} \sur{Yuan}\,\orcidlink{0000-0001-5403-6576}}
\email{yuanyuan12430@126.com}
\author[3]{\fnm{Xingwu} \sur{Liu}}
\author[1]{\fnm{Xin} \sur{Han}\,\orcidlink{0000-0002-1694-7712}}
\author[1]{\fnm{Yong} \sur{Zhou}}

\affil*[1]{\orgdiv{School of Software},
  \orgname{Dalian University of Technology},
  \orgaddress{\city{Dalian}, \postcode{116620}, \state{Liaoning}, \country{China}}}

\affil[2]{\orgdiv{School of Information and Communication Engineering},
  \orgname{Dalian Minzu University},
  \orgaddress{\city{Dalian}, \postcode{116600}, \state{Liaoning}, \country{China}}}

\affil[3]{\orgdiv{School of Mathematical Sciences},
  \orgname{Dalian University of Technology},
  \orgaddress{\city{Dalian}, \postcode{116024}, \state{Liaoning}, \country{China}}}

\abstract{We consider the two-machine open-shop scheduling problem in which one machine is unavailable during a fixed interval. We study the resumable setting: an operation interrupted by the unavailability interval may resume, without penalty, when the machine becomes available. The objective is to minimize the makespan. Although the problem is NP-hard and several approximation algorithms are known, whether it admits a fully polynomial-time approximation scheme (FPTAS) has remained open for two decades. We resolve this question affirmatively by giving the first FPTAS, thereby strengthening the previously known polynomial-time approximation scheme (PTAS). As an intermediate result, we develop a new pseudo-polynomial dynamic program with seven state dimensions, improving on the ten-dimensional formulation in the literature.}


\maketitle

\section{Introduction}
	Machine-scheduling models with maintenance or availability constraints have received considerable attention in operations research. We study a two-machine open-shop problem in which one machine is subject to a maintenance period. This problem was introduced by Breit et al.~\cite{breit2001two}.
	
	Maintenance periods (MPs) are commonly classified as either \emph{fixed} or \emph{flexible}. In the fixed model, both endpoints of the MP are prescribed. In the flexible model, the MP has a fixed duration, but its starting time may be chosen within a given time window.
	
	An instance of the classical two-machine problem consists of a set $N=\{1,2,\ldots,n\}$ of jobs to be processed on machines $M_1$ and $M_2$. Each job $i\in N$ comprises two operations, $O_{i,1}$ and $O_{i,2}$, with processing times $a_i$ on $M_1$ and $b_i$ on $M_2$, respectively. Operations of the same job may not overlap, and each machine can process at most one operation at a time. Apart from an interruption caused by the MP in the resumable model defined below, operations are nonpreemptive. If $O_{i,1}$ must precede $O_{i,2}$ for every job $i$, the environment is a flow shop; if the order of the two operations may be chosen independently for each job, it is an open shop. The classical two-machine open-shop problem, denoted by $O2\Vert C_{\max}$, can be solved in $O(n)$ time by the algorithm of Gonzalez and Sahni~\cite{gonzalez1976open}; alternative linear-time algorithms appear in~\cite{pinedo1982stochastic,de1989graph}. The objective is to minimize the \emph{makespan}, that is, the latest completion time among all operations.
	
	We extend this classical problem by imposing a fixed MP on one of the machines. The MP occupies the interval $[s,e]$, during which the machine is unavailable.
	
	Three standard models describe how an operation interacts with a fixed MP. In the \emph{resumable} model, an operation interrupted by the MP resumes without penalty when the machine becomes available~\cite{lee1996machine}. In the \emph{non-resumable} model, such an operation must restart from the beginning~\cite{lee1996machine}. In the \emph{semi-resumable} model, it must repeat part of its processing after the MP~\cite{lee1999two}.
	
	The two-machine open-shop problem with a single fixed MP is NP-hard~\cite{breit2001two}, which motivates the study of approximation algorithms. An algorithm is a $\rho$-approximation, for $\rho\geq 1$, if it returns a schedule whose makespan is at most $\rho$ times the optimum. A polynomial-time approximation scheme (PTAS) achieves ratio $1+\varepsilon$ for every fixed $\varepsilon>0$ in time polynomial in the input size. If the running time is also polynomial in $1/\varepsilon$, the scheme is a fully polynomial-time approximation scheme (FPTAS).
	
	Both exact and approximation algorithms are known for two-machine open-shop problems with fixed MPs. In the resumable setting, Breit et al.~\cite{breit2001two} gave a $4/3$-approximation algorithm for a single fixed MP. Kubzin et al.~\cite{kubzin2006polynomial} subsequently developed two PTASs: one for a single fixed MP on each machine and one for multiple fixed MPs on one machine. Lorigeon et al.~\cite{lorigeon2002dynamic} gave a pseudo-polynomial dynamic program for a single fixed MP. In the non-resumable setting, Breit et al.~\cite{breit2003non} obtained a $4/3$-approximation algorithm, and Yuan et al.~\cite{yuan2021ptas} later developed a PTAS.
	
	The corresponding two-machine flow-shop problem with a single fixed MP is NP-hard in the resumable setting~\cite{lee1997minimizing}. Lee~\cite{lee1997minimizing} gave approximation algorithms with ratios $3/2$ and $4/3$ when the MP lies on $M_1$ and $M_2$, respectively, and Ng and Kovalyov~\cite{ng2004fptas} subsequently developed an FPTAS. For the non-resumable setting with a single fixed MP on $M_1$, Lu et al.~\cite{lu2026flowshop} recently presented the first FPTAS. The structural similarity between the flow-shop and open-shop variants, together with the pseudo-polynomial algorithm of Lorigeon et al.~\cite{lorigeon2002dynamic}, suggests that the open-shop variant may also admit an FPTAS.
	
	We confirm this conjecture. For the resumable problem, denoted by $O2\mid r\text{-}a(M_1)\mid C_{\max}$~\cite{graham1979optimization}, we decompose a candidate schedule into three two-machine flow-shop subproblems. Their state information suffices to recover the makespan of the original instance. This decomposition yields a seven-dimensional pseudo-polynomial dynamic program with running time $O(nsV^5)$, where $V$ is the total processing time of all jobs. Unlike the earlier formulation, our dynamic program minimizes $a(XP)$ evaluated using the original processing times; retaining this value after rounding preserves feasibility with respect to $s$. Applying standard scaling and rounding then yields the first FPTAS for the problem, with running time $O(n^7/\varepsilon^6)$.
	
	The remainder of the paper is organized as follows. Section~\ref{Pre} establishes the structural properties used throughout the paper. Section~\ref{DP} presents the pseudo-polynomial dynamic program, and Section~\ref{FPTAS} derives the FPTAS. Section~\ref{Conclusion} concludes the paper.
	
	\section{Preliminaries}\label{Pre}
	
	We first introduce notation and establish structural properties of optimal schedules for $O2\mid r\text{-}a(M_1)\mid C_{\max}$. In particular, the jobs can be partitioned into four subsets according to their operation order and their position relative to the MP on $M_1$. This structure underlies our dynamic program.
	
	\noindent\textbf{Johnson's rule (JR).}
	Job $J_i$ precedes job $J_j$ if $\min\{a_i,b_j\}<\min\{a_j,b_i\}$; ties are broken arbitrarily~\cite{johnson1954optimal}. For the two-machine flow-shop problem, JR produces an optimal permutation schedule.
	
	Let $N=\{1,2,\ldots,n\}$ denote the job set. For each job $i\in N$, let $a_i$ and $b_i$ be its processing times on $M_1$ and $M_2$, respectively. For any $Q\subseteq N$, define $a(Q)=\sum_{i\in Q}a_i$ and $b(Q)=\sum_{i\in Q}b_i$. Let $C_{\max}(\pi)$ denote the makespan of a schedule $\pi$. If $\widehat{S}^*$ is an optimal schedule for $O2\Vert C_{\max}$, then~\cite{gonzalez1976open}
	\begin{equation}
		C_{\max}(\widehat{S}^*) = \max\left\{a(N),\ b(N),\ \max_{i\in N}\{a_i+b_i\}\right\}.
	\end{equation}
	
	The MP is the interval $[s,e]$, of duration $d=e-s$. Without loss of generality, it is imposed on $M_1$, since the case in which it is imposed on $M_2$ is symmetric. We assume throughout that $a_i$, $b_i$, $s$, and $d$ are positive integers encoded in binary. We restrict attention to $0<s<a(N)$. The boundary case $s=0$ is solvable in $O(n)$ time by the algorithm of Lu and Zhang~\cite{lu1993np}, whereas the case $s\geq a(N)$ is solvable in $O(n)$ time by the algorithm of Shakhlevich et al.~\cite{shakhlevich1993two}.
	
	Let $X$ be the set of jobs processed first on $M_1$ and then on $M_2$, and let $Y$ be the set processed in the reverse order. Let $t_{i,j}$ denote the starting time of job $i$ on machine $M_j$. We define
	
	\begin{align*}
		XP &= \{i \in X \mid t_{i,1} + a_i \leq s\}, \quad XS = \{i \in X \mid t_{i,1} + a_i > s\}, \\
		YP &= \{i \in Y \mid t_{i,1} < s\}, \quad YS = \{i \in Y \mid t_{i,1} \geq s\}.
	\end{align*}

	\Needspace{10\baselineskip}
	The following lemma identifies the two canonical forms on which our algorithm is based.
	\begin{lemma}\label{lem:canonical}
		There exists an optimal schedule of one of the following two forms:
		\begin{enumerate}
			\item On \( M_1 \), the job sequence is \( XP / XS / YS \); on \( M_2 \), the job sequence is \( YS / XP / XS \) (see Figure \ref{fig1:suba});
			\item On \( M_1 \), the job sequence is \( XP / YP / YS \); on \( M_2 \), the job sequence is \( YP / YS / XP \) (see Figure \ref{fig1:subb}).
		\end{enumerate}
		In either form, the jobs within each subset appear in the same order on both machines.
	\end{lemma}

	\begin{figure}[ht]
		\centering
		\begin{subfigure}[b]{0.45\textwidth}
			\centering
			\includegraphics[width=\textwidth]{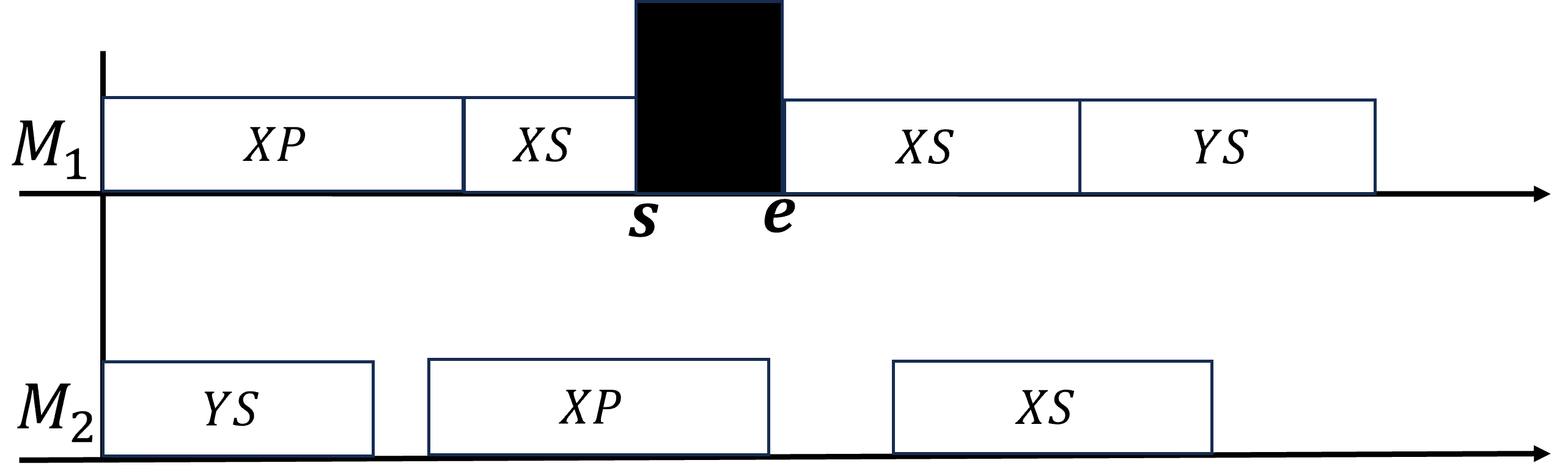}
			\caption{Form 1}
			\label{fig1:suba}
		\end{subfigure}
		\hfill
		\begin{subfigure}[b]{0.45\textwidth}
			\centering
			\includegraphics[width=\textwidth]{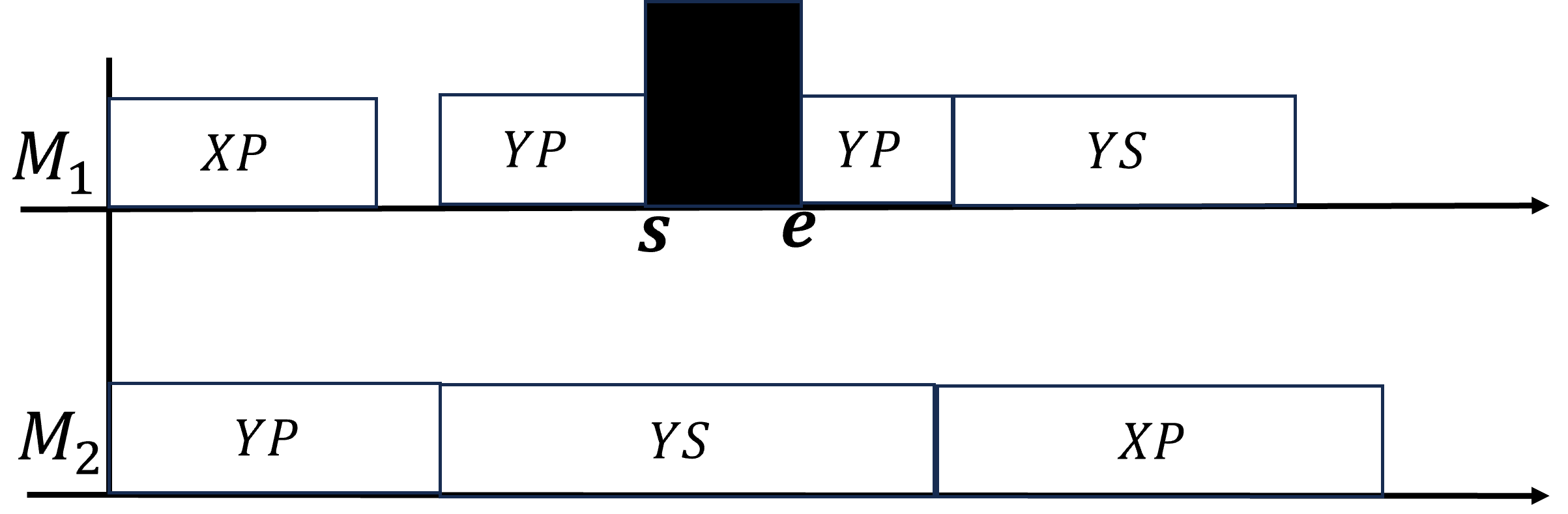}
			\caption{Form 2}
			\label{fig1:subb}
		\end{subfigure}
		\caption{Two canonical forms of an optimal schedule.}
		\label{fig1:main}
	\end{figure}

	\begin{proof}
		We establish the result through two exchange claims.

		\noindent\textit{Claim 1. There exists an optimal schedule in which at least one of $YP$ and $XS$ is empty.}

		Suppose that both $YP$ and $XS$ are nonempty, as in Figure~\ref{fig15}. Repeatedly interchange jobs in $XS$ and $YP$ until either every job formerly in $YP$ belongs to $YS$ or every job formerly in $XS$ belongs to $XP$. These two outcomes are illustrated in Figure~\ref{fig16}. The interchanges do not increase the makespan, so one of the resulting schedules is optimal and has $YP=\emptyset$ or $XS=\emptyset$.
		\begin{figure}[ht]
			\centering
			\includegraphics[width=0.6\textwidth]{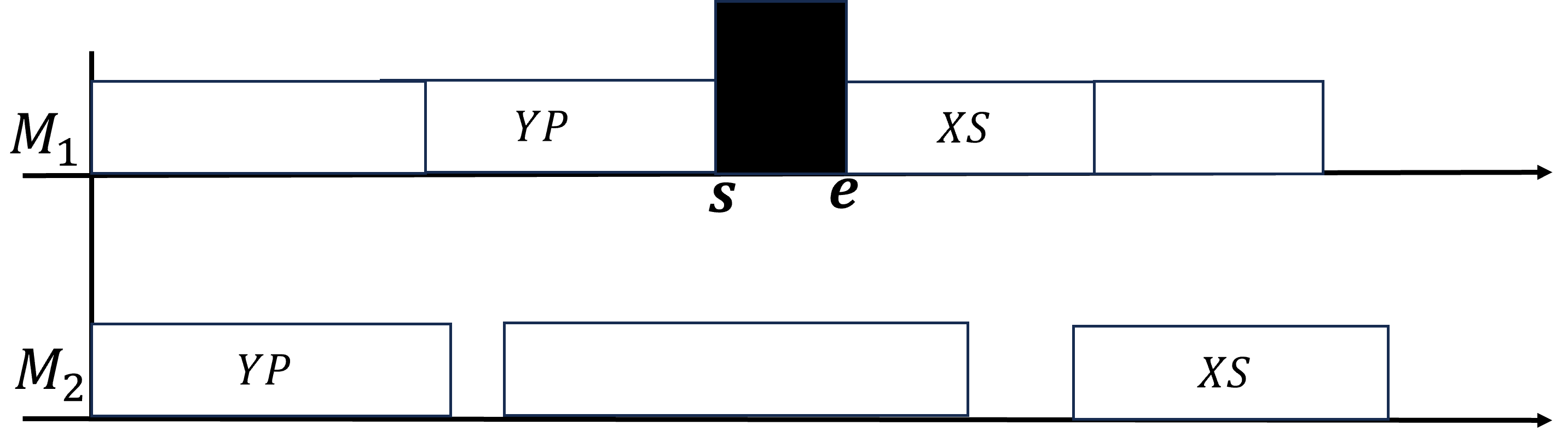}
			\caption{A schedule in which both $YP$ and $XS$ are nonempty.}
			\label{fig15}
		\end{figure}
		\begin{figure}[ht]
			\centering
			\begin{subfigure}[b]{0.45\textwidth}
				\centering
				\includegraphics[width=\textwidth]{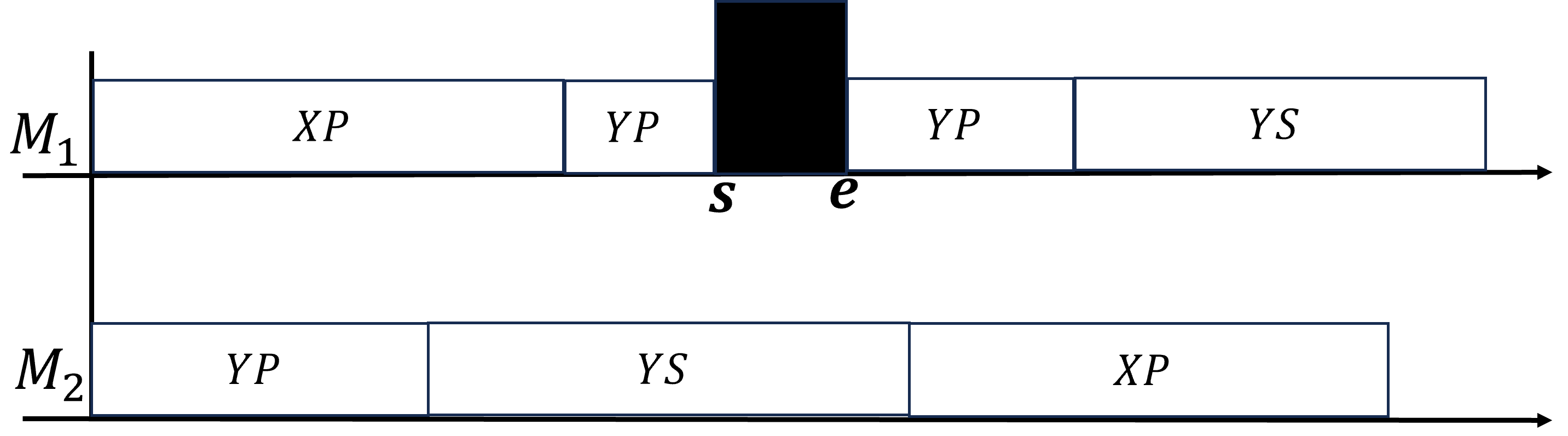}
				\caption{All jobs in $XS$ are moved to $XP$.}
				\label{fig16:suba}
			\end{subfigure}
			\hfill
			\begin{subfigure}[b]{0.45\textwidth}
				\centering
				\includegraphics[width=\textwidth]{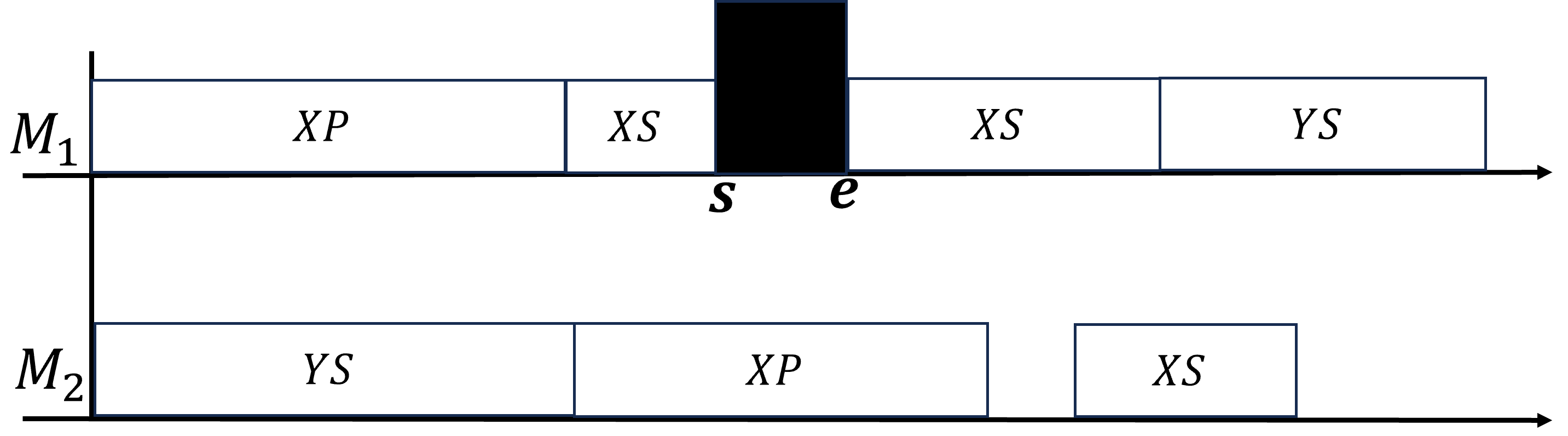}
				\caption{All jobs in $YP$ are moved to $YS$.}
				\label{fig16:subb}
			\end{subfigure}
			\caption{Interchanges between jobs in $XS$ and $YP$.}
			\label{fig16}
		\end{figure}

		\medskip
		\noindent\textit{Claim 2. There exists an optimal schedule in which the subset order is $XP/YP/XS/YS$ on $M_1$ and $YP/YS/XP/XS$ on $M_2$. Moreover, within each subset, the jobs appear in the same order on both machines.}

		Consider jobs $J_{yp}\in YP$ and $J_{xp}\in XP$ such that $J_{yp}$ precedes $J_{xp}$, as in Figure~\ref{fig17:suba}. Interchanging these jobs does not increase the makespan (Figure~\ref{fig17:subb}). The same argument applies to jobs $J_{ys}\in YS$ and $J_{xs}\in XS$ whenever $J_{ys}$ precedes $J_{xs}$. Repeating these interchanges yields the subset order $XP/YP/XS/YS$ on $M_1$. A symmetric argument gives the order $YP/YS/XP/XS$ on $M_2$. Within each subset, the same pairwise-interchange argument makes the job orders on the two machines identical.
	\begin{figure}[ht]
		\centering
		\includegraphics[width=0.6\textwidth]{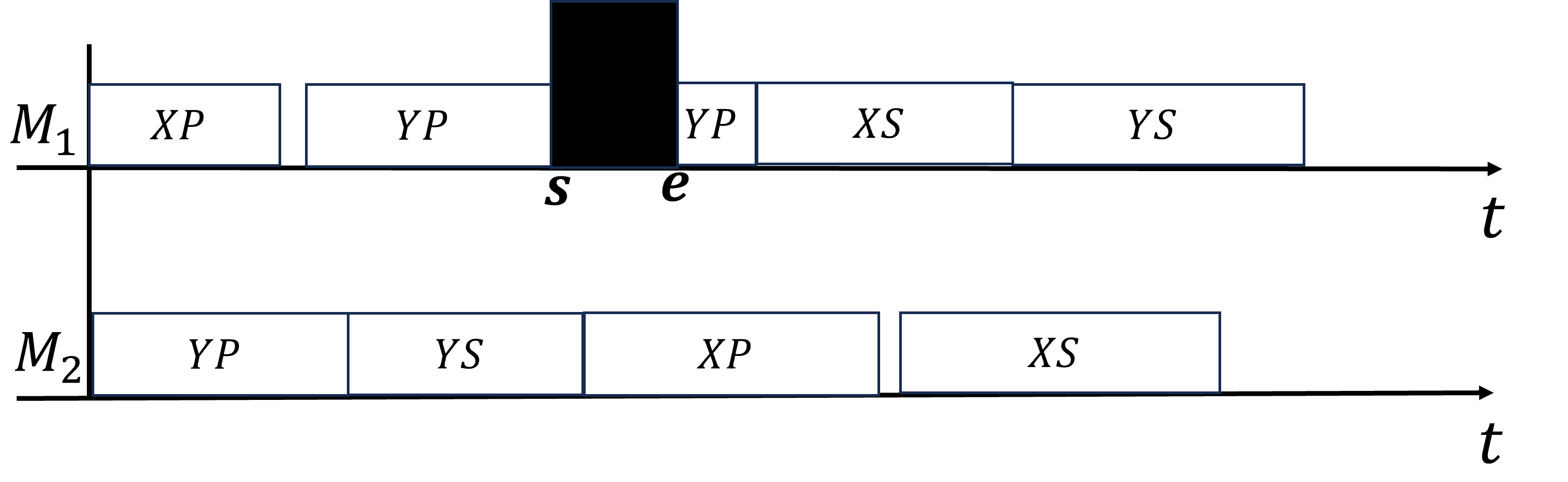}
		\caption{A dominant schedule for $O2\mid r\text{-}a(M_1)\mid C_{\max}$.}
		\label{fig:p1}
	\end{figure}
	\begin{figure}[ht]
		\centering
		\begin{subfigure}[b]{0.45\textwidth}
			\centering
			\includegraphics[width=\textwidth]{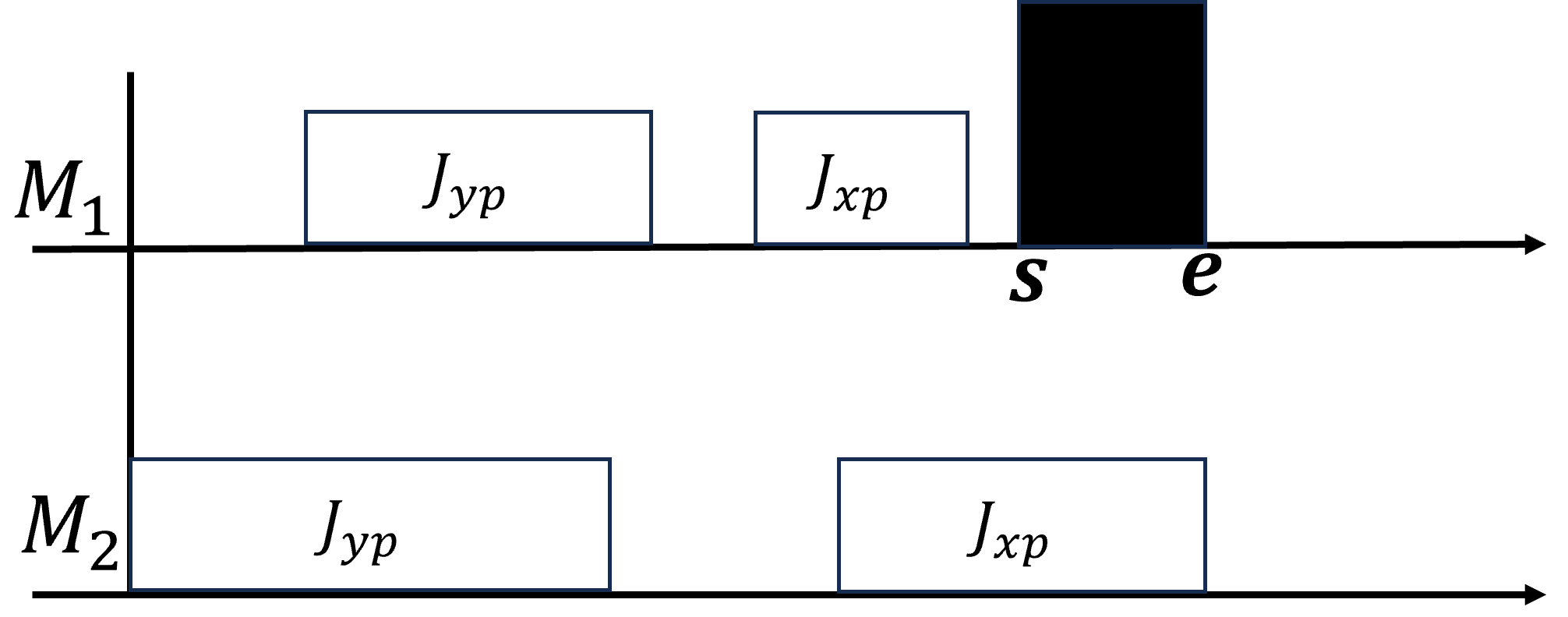}
			\caption{$J_{yp}$ precedes $J_{xp}$.}
			\label{fig17:suba}
		\end{subfigure}
		\hfill
		\begin{subfigure}[b]{0.45\textwidth}
			\centering
			\includegraphics[width=\textwidth]{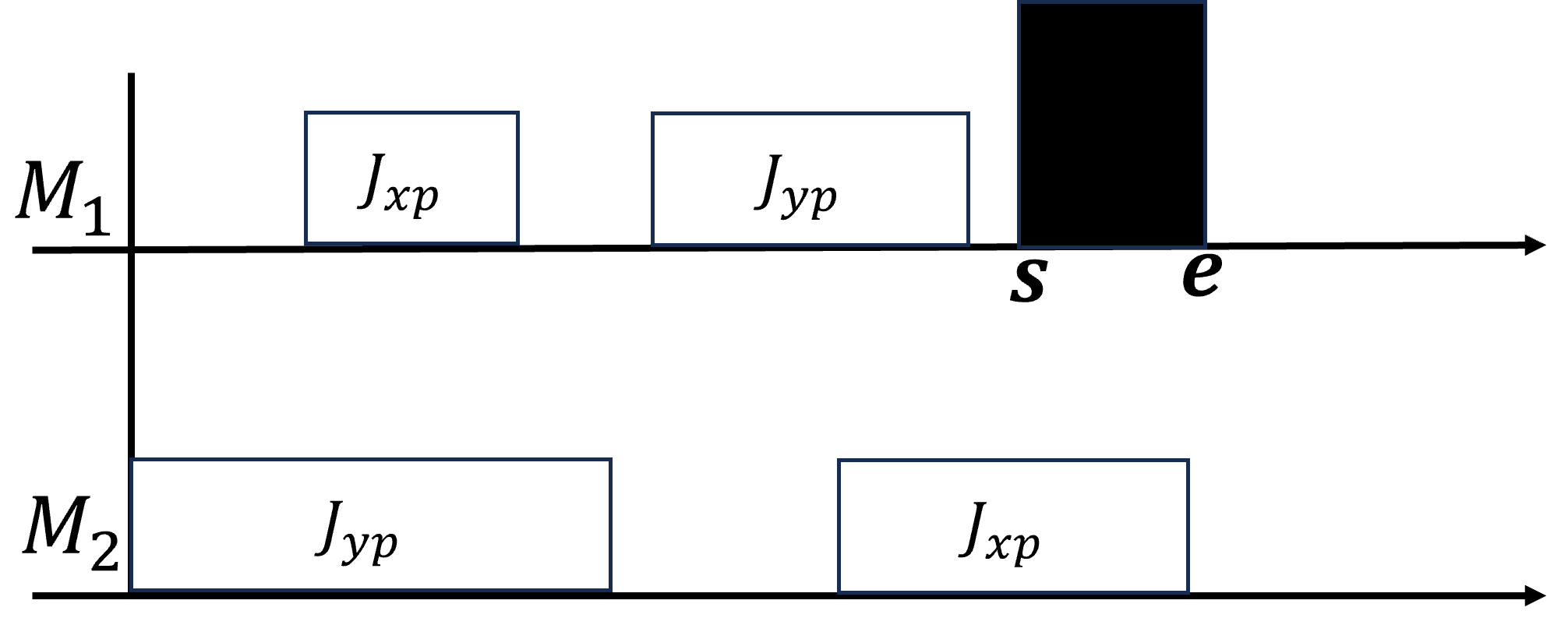}
			\caption{The schedule after interchanging $J_{yp}$ and $J_{xp}$.}
			\label{fig17:subb}
		\end{subfigure}
		\caption{The interchange used in the proof of Claim~2.}
		\label{fig17}
	\end{figure}

		By Claim~1, we may restrict attention to schedules with $YP=\emptyset$ or $XS=\emptyset$. Applying Claim~2 to these two cases gives, respectively, the two forms stated in the lemma.
	\end{proof}

The same structural claim appears in~\cite{lorigeon2002dynamic}, but with a different partition. Here, membership in $YP$ and $YS$ is determined by the starting time on $M_1$, whereas Lorigeon et al.~classify these jobs by their starting time on $M_2$; specifically, they define
	\begin{align*}
		YP &= \{i \in Y \mid t_{i,2} < s\}, \quad YS = \{i \in Y \mid t_{i,2} \geq s\}.
	\end{align*}
	
	This distinction is essential: Lemma~\ref{lem:canonical} does not hold under that alternative definition. The two-job instance in Figure~\ref{fig14:main} has a unique optimal schedule. Its job order is $J_1,J_2$ (that is, $XS/YS$) on $M_1$ and $J_2,J_1$ (that is, $YS/XS$) on $M_2$. Under the classification in~\cite{lorigeon2002dynamic}, however, these orders would be $XS/YP$ on $M_1$ and $YP/XS$ on $M_2$.

\begin{figure}[ht]
		\centering
		\includegraphics[width=0.6\textwidth]{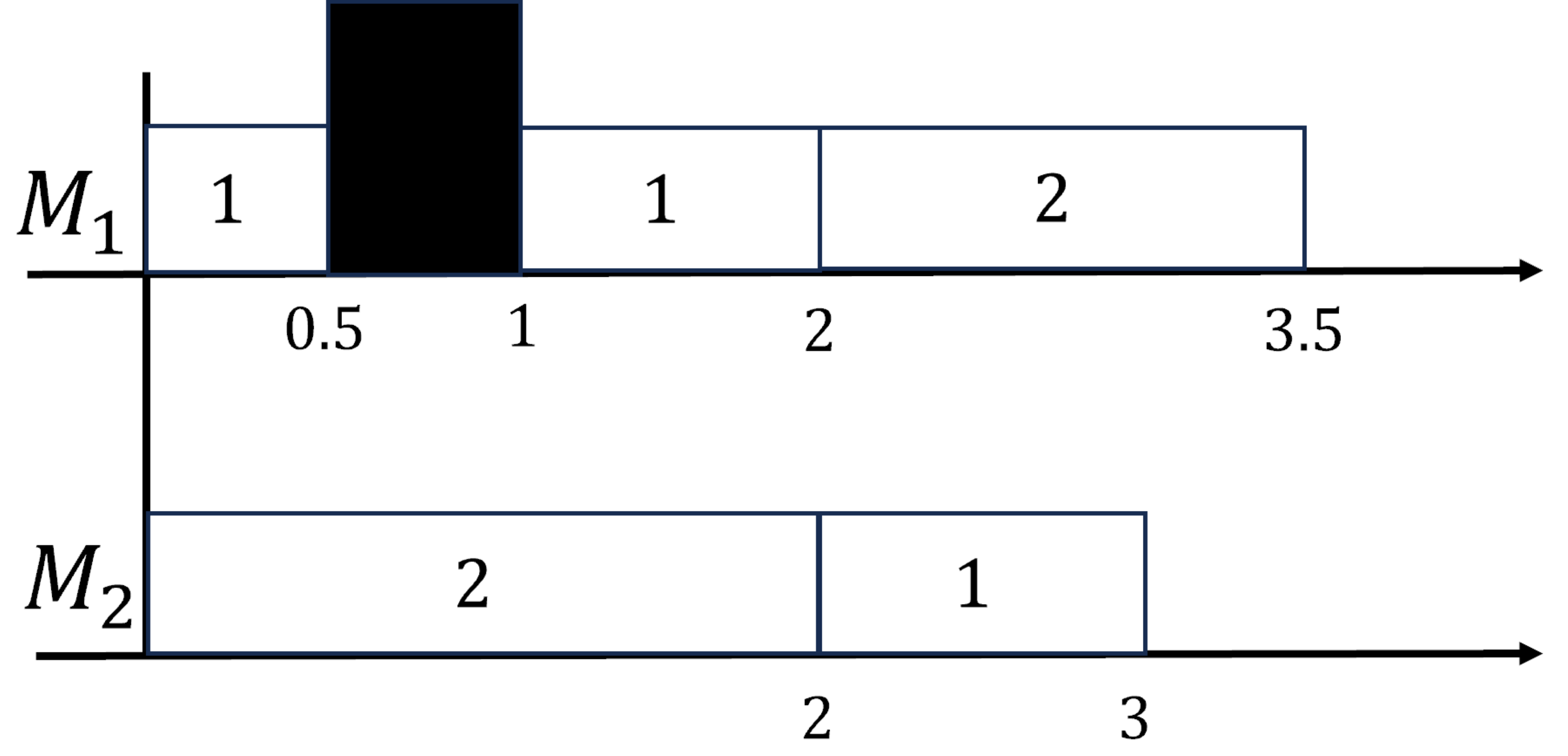}
		\caption{$a_1 = 1.5, b_1 = 1, a_2 = 1.5, b_2 = 2, s = 0.5, e = 1.$}
		\label{fig14:main}
	\end{figure}
	
	The following ordering property completes the structural description needed by the dynamic program.
	\begin{lemma}\label{lem:johnson-orders}
		An optimal canonical schedule can be chosen so that the jobs in $XP$ and $XS$ follow JR, whereas those in $YP$ and $YS$ follow the reverse of JR.
	\end{lemma}
	\begin{proof}
		By Lemma~\ref{lem:canonical}, we may start from an optimal schedule in which the job order is identical on both machines within each subset. Once the subset boundaries are fixed, $XP$ and $XS$ each form an ordinary two-machine flow-shop subproblem, so Johnson's pairwise-interchange argument gives JR without increasing the makespan. For $YP$ and $YS$, the machine order is reversed; applying the same argument after exchanging the roles of $M_1$ and $M_2$ gives the reverse of JR. These interchanges preserve the subset boundaries and hence preserve the corresponding canonical form.
	\end{proof}
	
	\section{A pseudo-polynomial dynamic program}\label{DP}
	
	We develop a pseudo-polynomial dynamic program for $O2\mid r\text{-}a(M_1)\mid C_{\max}$ with running time $O(nsV^5)$, where $V=a(N)+b(N)$. The earlier dynamic program of Lorigeon et al.~\cite{lorigeon2002dynamic} uses ten state dimensions and runs in $O(ns^2V^7)$ time, whereas our formulation uses seven dimensions. A second important distinction is the objective stored by the dynamic program. Rather than minimizing the makespan directly, we minimize $a(XP)$ evaluated using the original processing times. Retaining this unrounded value in the rounded dynamic program preserves feasibility with respect to $s$, while the makespan remains recoverable from the state variables. Consequently, the formulation provides a suitable basis for an FPTAS.
	
	\subsection{Computing the makespan from \texorpdfstring{\(XP\), \(XS\), \(YP\), and \(YS\)}{XP, XS, YP, and YS}}
	By Lemmas~\ref{lem:canonical} and~\ref{lem:johnson-orders}, an optimal schedule can be obtained by considering the two canonical forms and applying the prescribed Johnson order within each subset. It therefore remains to determine the makespan induced by a fixed partition.
	
	Let $C_{\max}(XP)$ and $C_{\max}(XS)$ denote the makespans obtained by applying JR to $XP$ and $XS$, respectively. Similarly, let $C_{\max}(YP)$ and $C_{\max}(YS)$ denote the makespans obtained by applying reverse JR to $YP$ and $YS$. Each subset is viewed as an independent two-machine flow-shop instance without an MP.
	\begin{proposition}\label{prop:makespan-a}
		For a schedule with $YP=\emptyset$, the makespan is
		\begin{equation}\label{eq2}
			\begin{split}
				C_{\max}= \max\bigl\{
				& a(N) + d, \, C_{\max}(YS), \, e+a(YS),  \\
				& C_{\max}(XP) + b(XS), \, a(XP) + C_{\max}(XS) + d , \, b(N) \bigr\}.
			\end{split}
		\end{equation}
	\end{proposition}
	
	\begin{proof}
		The last operation completes either on $M_1$ or on $M_2$. These two cases are characterized by Observations~\ref{obs1} and~\ref{obs2}, respectively.
		
		\begin{observation}\label{obs1}
			If the last operation completes on $M_1$, then
			\begin{align*}
				C_{\max}=\max\{a(N)+d,C_{\max}(YS), e+a(YS)\}.
			\end{align*}
		\end{observation}
		
		Suppose first that $XS\neq\emptyset$. Apart from the MP, $XP$ and $XS$ are processed consecutively on $M_1$, and the only possible subsequent idle interval lies between $XS$ and $YS$. If this interval is absent, the bound $a(N)+d$ is attained; otherwise, the bound $C_{\max}(YS)$ is attained (Figure~\ref{fig7}). In the degenerate case $XS=\emptyset$, we have $a(XP)\leq s$, and hence $a(N)+d\leq e+a(YS)$. The completion time of $M_1$ is then $\max\{e+a(YS),C_{\max}(YS)\}$, so the same formula remains valid (Figure~\ref{fig18}).
		
		\begin{figure}[ht]
			\centering
			\begin{subfigure}[b]{0.45\textwidth}
				\centering
				\includegraphics[width=\textwidth]{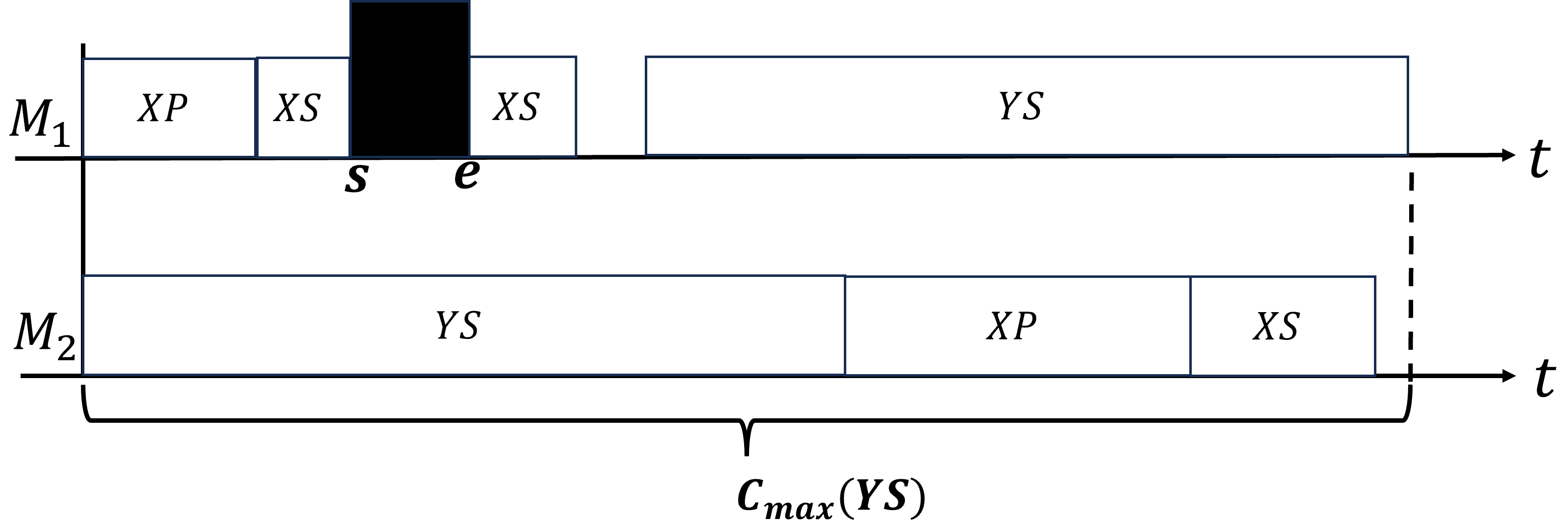}
				\caption{An idle interval between $XS$ and $YS$ on $M_1$.}
				\label{fig7}
			\end{subfigure}
			\hfill
			\begin{subfigure}[b]{0.45\textwidth}
				\centering
				\includegraphics[width=\textwidth]{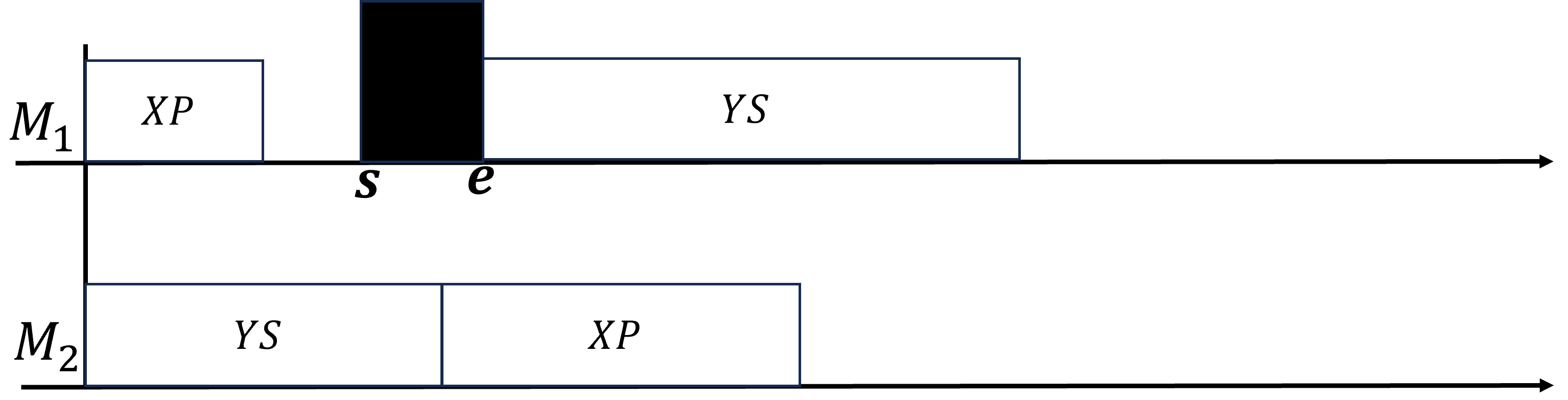}
				\caption{$XS = \emptyset$.}
				\label{fig18}
			\end{subfigure}
			\caption{The last operation completes on $M_1$.}
			\label{fig7_18}
		\end{figure}
		
		\begin{observation}\label{obs2}
			If the last operation completes on $M_2$, then
			\begin{align*}
				C_{\max}=\max\{C_{\max}(XP)+b(XS),a(XP)+d+C_{\max}(XS),b(N)\}.
			\end{align*}
		\end{observation}
		
		Assume first that $XS\neq\emptyset$. The three terms in Observation~\ref{obs2} correspond to the three possible idle-time patterns on $M_2$. If an idle interval occurs between $YS$ and $XP$, but $XP$ and $XS$ are consecutive, then $C_{\max}(XP)+b(XS)$ is attained (Figure~\ref{fig8}). If an idle interval occurs between $XP$ and $XS$, then $a(XP)+d+C_{\max}(XS)$ is attained, regardless of whether an earlier idle interval also occurs (Figure~\ref{fig9}). If neither interval occurs, $M_2$ is continuously busy and completes at time $b(N)$.
		
		\begin{figure}[ht]
			\centering
			\begin{subfigure}[b]{0.45\textwidth}
				\centering
				\includegraphics[width=\textwidth]{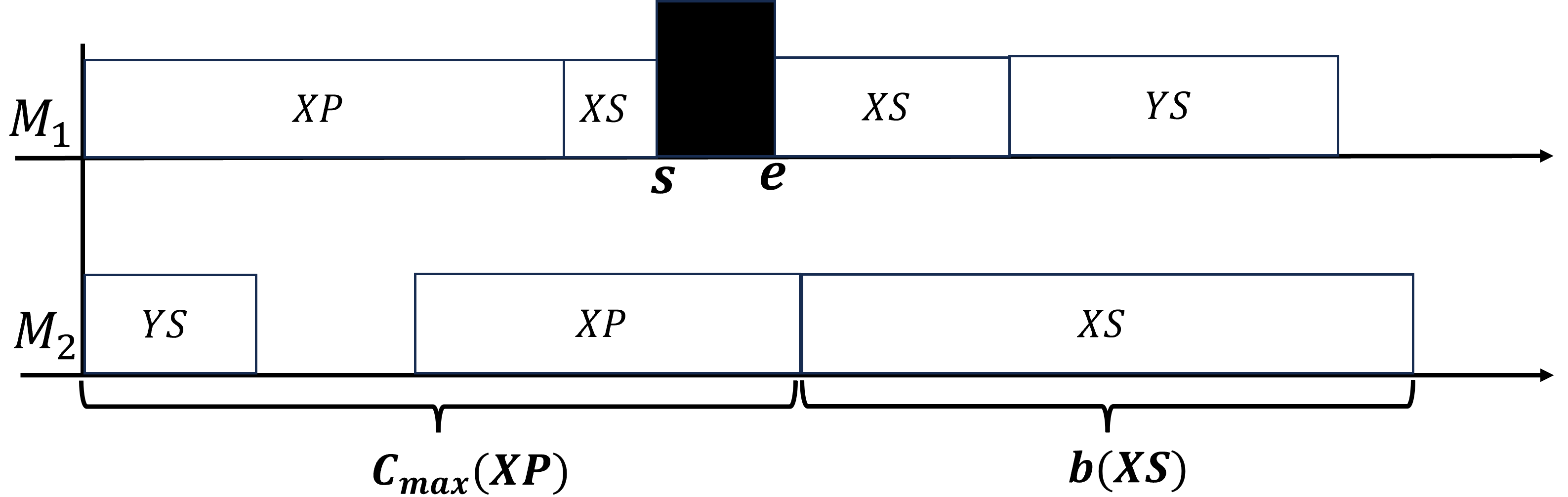}
				\caption{No idle interval between $XP$ and $XS$ on $M_2$.}
				\label{fig8}
			\end{subfigure}
			\hfill
			\begin{subfigure}[b]{0.45\textwidth}
				\centering
				\includegraphics[width=\textwidth]{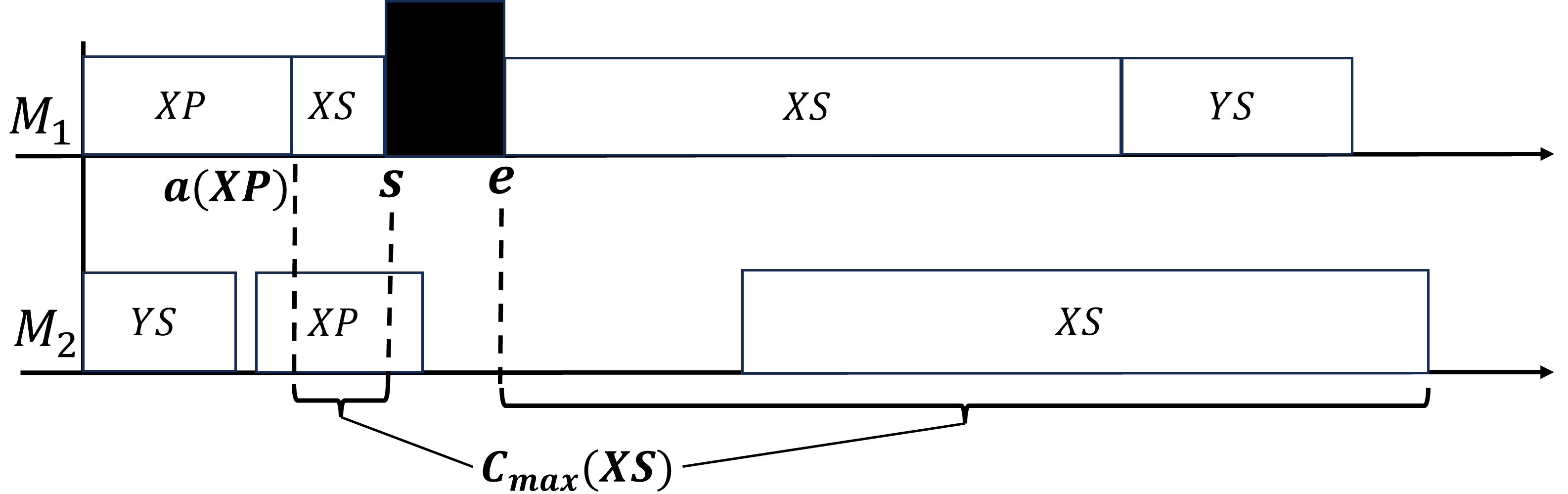}
				\caption{An idle interval between $XP$ and $XS$ on $M_2$, $XS \neq \emptyset$.}
				\label{fig9}
			\end{subfigure}
			\caption{The last operation completes on $M_2$.}
			\label{fig8_9}
		\end{figure}
		
		It remains to consider the degenerate case $XS=\emptyset$. Deleting the empty block leaves only one possible idle interval on $M_2$, namely, the interval between $YS$ and $XP$. If this interval is absent, $M_2$ completes at time $b(N)$ (Figure~\ref{fig19}); if it is present, $M_2$ completes at time $C_{\max}(XP)$ (Figure~\ref{fig20}). Moreover, $a(XP)\leq s$ implies $a(XP)+d\leq e$. Because the last operation is assumed to complete on $M_2$, its completion time is at least that of $M_1$, and hence at least $e$. Thus, the term $a(XP)+d+C_{\max}(XS)=a(XP)+d$ is redundant in this case, and Observation~\ref{obs2} remains valid without a separate formula.
		\begin{figure}[ht]
			\centering
			\begin{subfigure}[b]{0.45\textwidth}
				\centering
				\includegraphics[width=\textwidth]{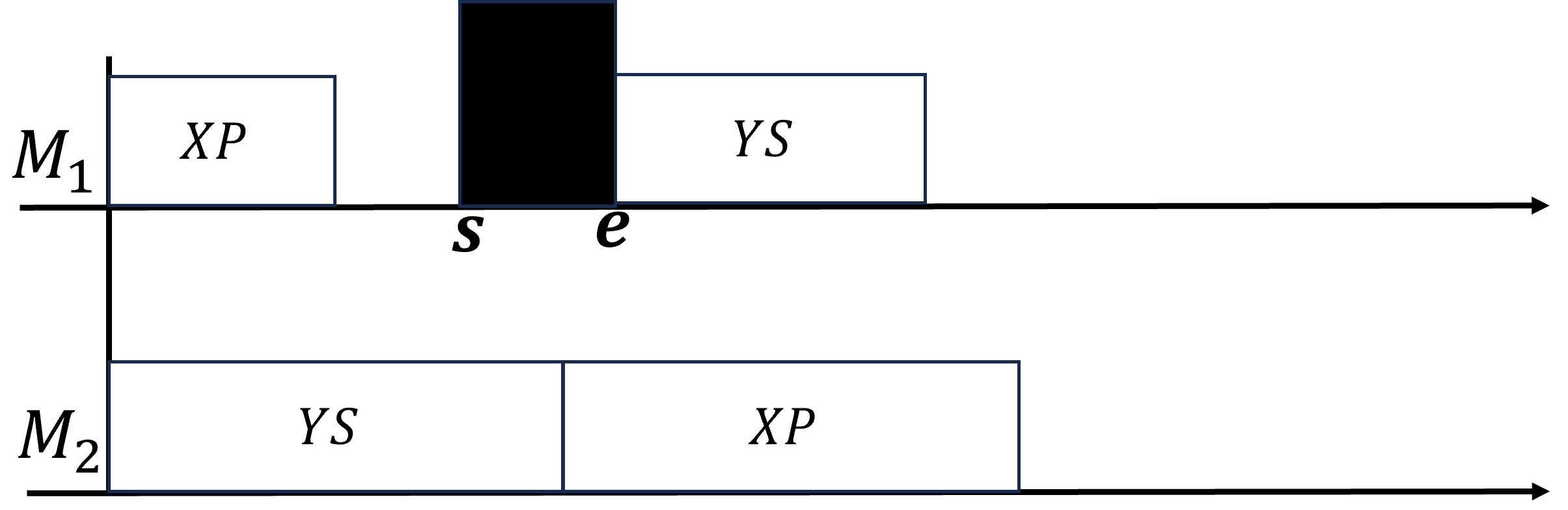}
				\caption{No idle interval between $YS$ and $XP$ on $M_2$.}
				\label{fig19}
			\end{subfigure}
			\hfill
			\begin{subfigure}[b]{0.45\textwidth}
				\centering
				\includegraphics[width=\textwidth]{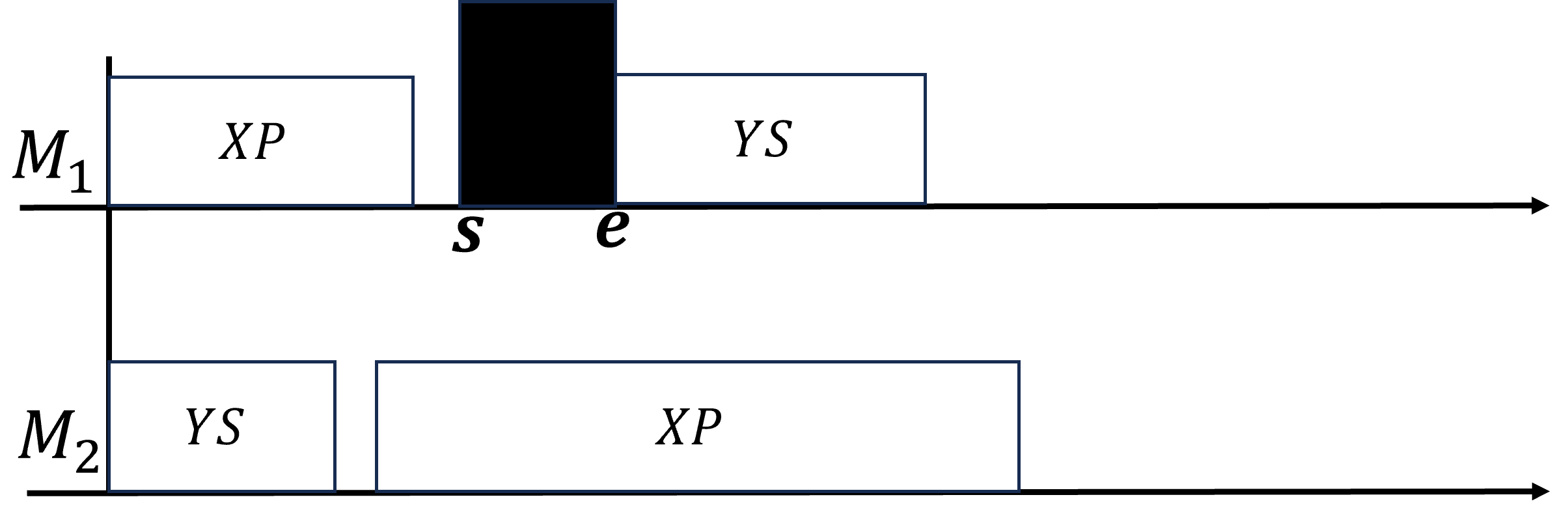}
				\caption{An idle interval between $YS$ and $XP$ on $M_2$.}
				\label{fig20}
			\end{subfigure}
			\caption{The last operation completes on $M_2$, $XS = \emptyset$.}
			\label{fig19_20}
		\end{figure}
		
		Combining these cases gives
		\[
		\begin{split}
			C_{\max}= \max\bigl\{
			& a(N) + d, \, C_{\max}(YS), \, e + a(YS),  \\
			& C_{\max}(XP) + b(XS), \, a(XP) + C_{\max}(XS) + d , \, b(N) \bigr\}.
		\end{split}
		\]
		
	\end{proof}
	\begin{proposition}\label{prop:makespan-b}
		For a schedule with $XS=\emptyset$, the makespan is
		\begin{equation}\label{eq3}
			\begin{split}
				C_{\max}= \max\bigl\{
				& b(YP) + C_{\max}(YS), \, C_{\max}(YP) + a(YS) + d, \\
				& a(N) + d, \, e+a(YS), \, b(N), \, C_{\max}(XP)\bigr\}.
			\end{split}
		\end{equation}
	\end{proposition}
	\begin{proof}
		The last operation completes either on $M_1$ or on $M_2$, as characterized by Observations~\ref{obs3} and~\ref{obs4}, respectively.

		\begin{observation}\label{obs3}
			If the last operation completes on $M_1$, then
			\begin{align*}
				C_{\max}=\max\{&C_{\max}(YP)+a(YS)+d,
				b(YP)+C_{\max}(YS),\\
				&a(N)+d,e+a(YS)\}.
			\end{align*}
		\end{observation}

		Suppose first that $YP\neq\emptyset$. Three cases arise on $M_1$. If there is no idle time between $YP$ and $YS$ but there is an idle interval between $XP$ and $YP$ (Figure~\ref{figa1}), then $C_{\max}=C_{\max}(YP)+a(YS)+d$. If there is an idle interval between $YP$ and $YS$, irrespective of whether one also occurs between $XP$ and $YP$ (Figure~\ref{figa2}), then $C_{\max}=b(YP)+C_{\max}(YS)$. If $M_1$ has no idle interval, then $C_{\max}=a(N)+d$.

		If $YP=\emptyset$, the subset orders reduce to $XP/YS$ on $M_1$ and $YS/XP$ on $M_2$, as in Figure~\ref{fig18}. Since every operation of $YS$ on $M_1$ starts no earlier than $e$, the MP contributes the term $e+a(YS)$. If precedence constraints within the reverse-JR schedule of $YS$ delay processing further, the relevant term is $C_{\max}(YS)$, which equals $b(YP)+C_{\max}(YS)$ because $b(YP)=0$. Consequently, Observation~\ref{obs3} also covers the case $YP=\emptyset$.

		\begin{figure}[ht]
			\centering
			\begin{subfigure}[b]{0.48\textwidth}
				\centering
				\includegraphics[height=2.35cm]{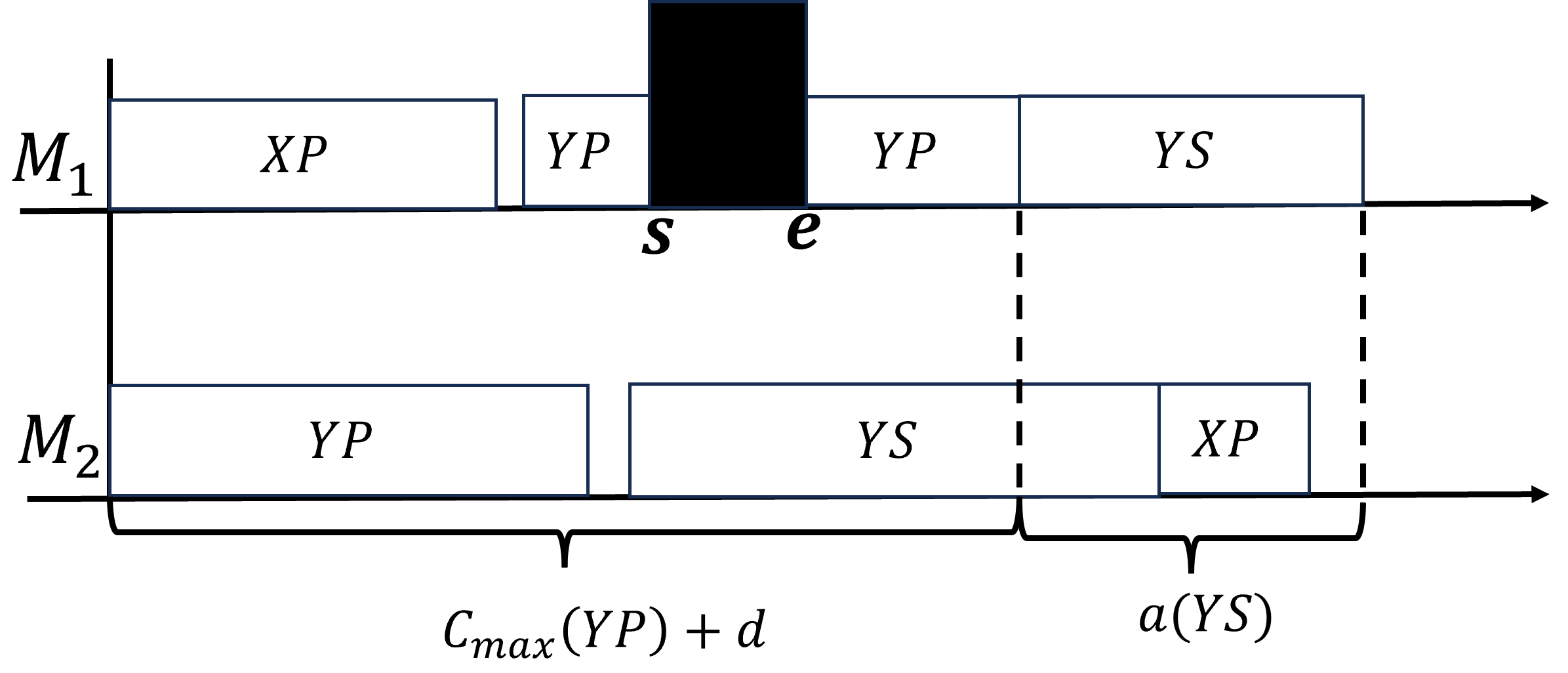}
				\caption{An idle interval between $XP$ and $YP$.}
				\label{figa1}
			\end{subfigure}
			\hfill
			\begin{subfigure}[b]{0.48\textwidth}
				\centering
				\includegraphics[height=2.35cm]{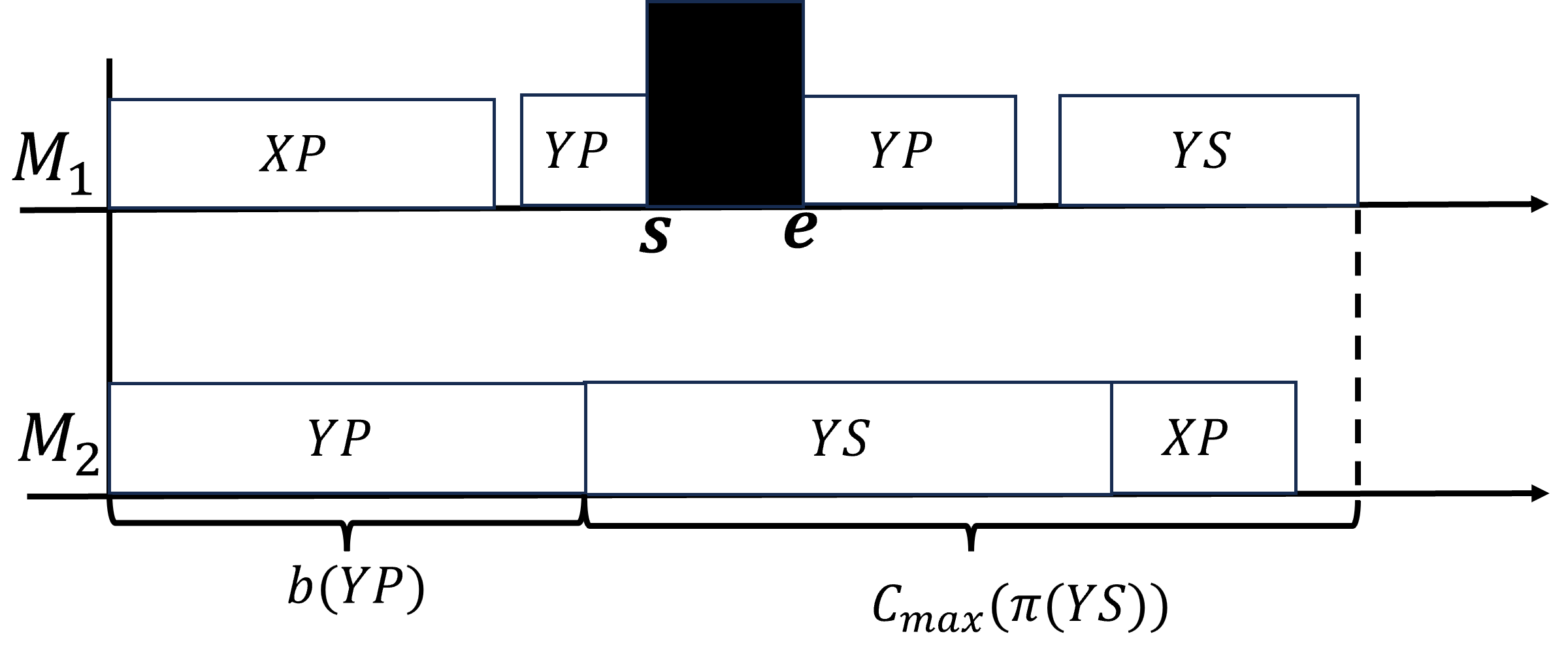}
				\caption{An idle interval between $YP$ and $YS$.}
				\label{figa2}
			\end{subfigure}
			\caption{Idle-time configurations when the last operation completes on $M_1$.}
			\label{figa1_a2}
		\end{figure}

		\begin{observation}\label{obs4}
			If the last operation completes on $M_2$, then
			\begin{align*}
				C_{\max}=\max\{b(N),C_{\max}(XP)\}.
			\end{align*}
		\end{observation}

		On $M_2$, no idle time occurs between $YP$ and $YS$, although an idle interval may occur between $YS$ and $XP$. If no such interval occurs, then $C_{\max}=b(N)$; otherwise, as in Figure~\ref{figa3}, $C_{\max}=C_{\max}(XP)$. This argument remains valid when $YP=\emptyset$, in which case $YS$ is simply the first subset processed on $M_2$.

		\begin{figure}[ht]
			\centering
			\includegraphics[width=0.8\textwidth]{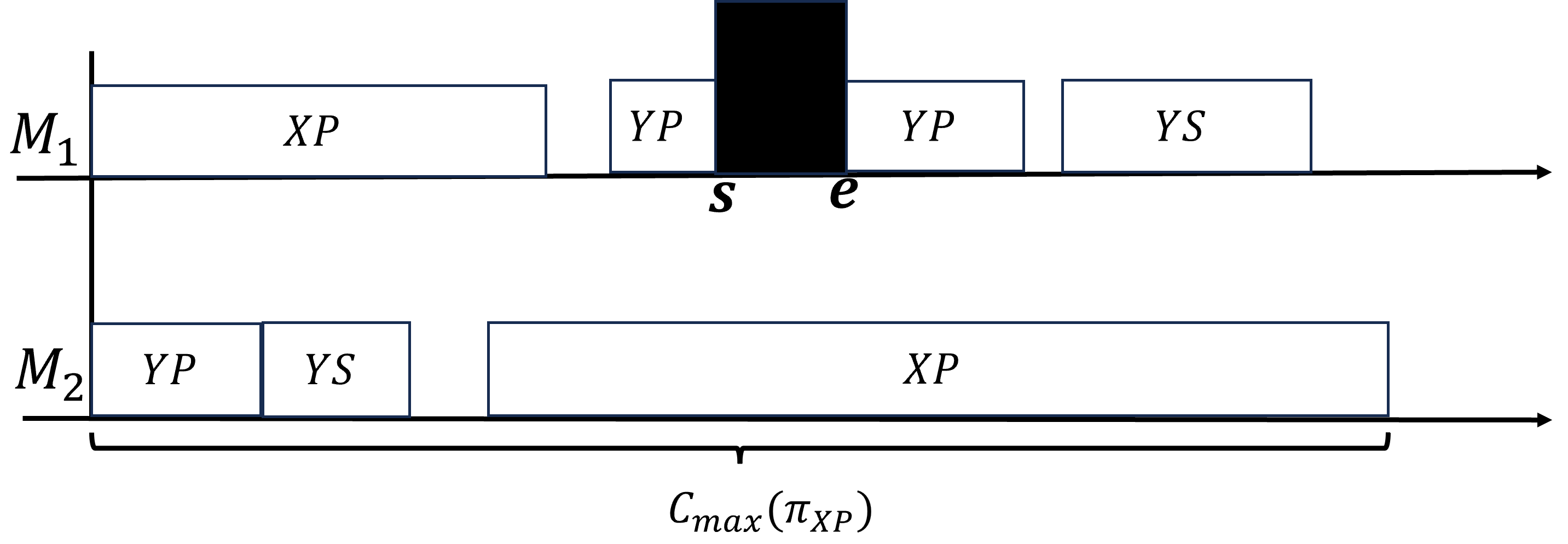}
			\caption{An idle interval between $YS$ and $XP$ on $M_2$.}
			\label{figa3}
		\end{figure}

		Combining these cases gives
		\begin{equation*}
			\begin{split}
				C_{\max}= \max\bigl\{
				& b(YP) + C_{\max}(YS), \, C_{\max}(YP) + a(YS) + d, \\
				& a(N) + d, \, e+a(YS), \, b(N), \, C_{\max}(XP)\bigr\},
			\end{split}
		\end{equation*}
		which is exactly~\eqref{eq3}.
	\end{proof}
	
	\subsection{Dynamic program}
	
	The dynamic program rests on three observations. First, each canonical schedule decomposes into three flow-shop subproblems whose schedules can be combined into an optimal schedule for the original instance. Each subproblem contributes two state variables, recording processing time on one of the machines or the subproblem makespan; together with the job index, this gives seven dimensions. Second, the value associated with a state records $a(N_k\cap XP)$ evaluated using the original processing times, rather than the makespan itself. Third, the subproblems omit the MP because Propositions~\ref{prop:makespan-a} and~\ref{prop:makespan-b} account for it explicitly when the global makespan is reconstructed.
	
	We first index the jobs in Johnson order. For each job $k$, the recurrence considers the three subsets to which $k$ may be assigned and updates the corresponding state variables. After the last job has been processed, we evaluate every feasible terminal state using Proposition~\ref{prop:makespan-a} or Proposition~\ref{prop:makespan-b} and recover a schedule by backtracking.

The formulation has two cases. Case A corresponds to Proposition~\ref{prop:makespan-a}, where $YP=\emptyset$, and Case B corresponds to Proposition~\ref{prop:makespan-b}, where $XS=\emptyset$.

	In the exact dynamic program, this value coincides with the state coordinate $u_1$, since both are computed from the original processing times. We nevertheless retain the value function explicitly because, in the rounded dynamic program, $u_1^0$ is computed from the scaled processing times whereas the value function continues to accumulate the original processing times $a_i$. This distinction ensures that feasibility with respect to the original bound $s$ can still be checked after rounding.
	
	\subsubsection{Case A: \texorpdfstring{\(YP=\emptyset\)}{YP = empty}}
	Index the jobs according to JR, and let $N_k$ denote the first $k$ jobs in this order.
	
	Let $f_k^a(u_1,u_2,u_3,v_1,v_2,v_3)$ be the minimum value of $a(N_k\cap XP)$, evaluated using the original processing times, among all partial assignments represented by the state. When the arguments are clear, we write simply $f_k^a$. The state variables are
	\begin{alignat*}{2}
		u_1 &= a(N_k \cap XP) & \; &  \text{total processing time on } M_1 \text{ of } N_k \cap XP; \\
		u_2 &= b(N_k \cap XS) & \; & \text{total processing time on } M_2 \text{ of } N_k \cap XS; \\
		u_3 &= a(N_k \cap YS) & \; & \text{total processing time on } M_1 \text{ of } N_k \cap YS; \\
		v_1 &= C_{\max}(\pi(N_k \cap XP)) & \; &  \text{makespan of } N_k \cap XP; \\
		v_2 &= C_{\max}(\pi(N_k \cap XS)) & \; & \text{makespan of } N_k \cap XS; \\
		v_3 &= C_{\max}(\pi(N_k \cap YS)) & \; & \text{makespan of } N_k \cap YS.
	\end{alignat*}
	
	Let $h_k = a(N_k) - u_1 - u_3$.
	
	\noindent\textit{Initial conditions.}
	\[
	f_k^a =
	\begin{cases}
		0,
		& \textrm{if} \quad (k, u_1, u_2, u_3, v_1, v_2, v_3) = (0, 0, 0, 0, *, *, *) \\[4pt]
		+\infty,
		& \textrm{otherwise}
	\end{cases}
	\]
	Here $*$ denotes an arbitrary nonnegative value.
	
	\noindent\textit{Recurrence.}
	\[
	f_k^a =
	\min\begin{cases}
		f_{k-1}^a (u_1-a_k, u_2, u_3, v_1-b_k, v_2, v_3) + a_k,
		& \textrm{if} \quad u_1 + b_k \leq v_1, \quad (a) \\[4pt]
		f_{k-1}^a (u_1, u_2-b_k, u_3, v_1, v_2-b_k, v_3),
		& \textrm{if} \quad h_k + b_k \leq v_2, \quad (b) \\[4pt]
		f_{k-1}^a (u_1, u_2, u_3-a_k, v_1, v_2, v_3-b_k),
		& \textrm{if} \quad u_3 + b_k \leq v_3, \quad (c) \\[4pt]
		+\infty,
		& \textrm{otherwise.}
	\end{cases}
	\]
	The ranges are $k\in\{0,1,\ldots,n\}$, $u_1\in[0,s]$, $u_2\in[0,b(N_k)]$, $u_3\in[0,a(N_k)]$, $v_1\in[u_1,V]$, $v_2\in[u_2,V]$, and $v_3\in[u_3,V]$.
	
	Transitions (a), (b), and (c) assign job $k$ to $XP$, $XS$, and $YS$, respectively.
	
	\begin{observation}\label{obs5}
		Every two-machine flow-shop instance admits an optimal schedule in which the only idle interval on $M_2$ precedes its first operation.
	\end{observation}
	
	\paragraph{Part (a): assigning job \texorpdfstring{$k$ to $XP$}{k to XP}.}
	Let $v_1'$ be the value of $v_1$ before job $k$ is assigned. If $u_1+b_k<v_1$, as in Figure~\ref{fig4:suba}, then $v_1'=v_1-b_k$. If $u_1+b_k=v_1$, as in Figure~\ref{fig4:subb}, Observation~\ref{obs5} allows the jobs preceding $k$ on $M_2$ to be delayed by $u_1$, again giving $v_1'=v_1-b_k$. Thus, in both cases (Figure~\ref{fig4}),
	\[
	f_k^a = f_{k-1}^a (u_1-a_k, u_2, u_3, v_1-b_k, v_2, v_3) + a_k.
	\]
	
	\begin{figure}[ht]
		\centering
		\begin{subfigure}[b]{0.45\textwidth}
			\centering
			\includegraphics[width=\textwidth]{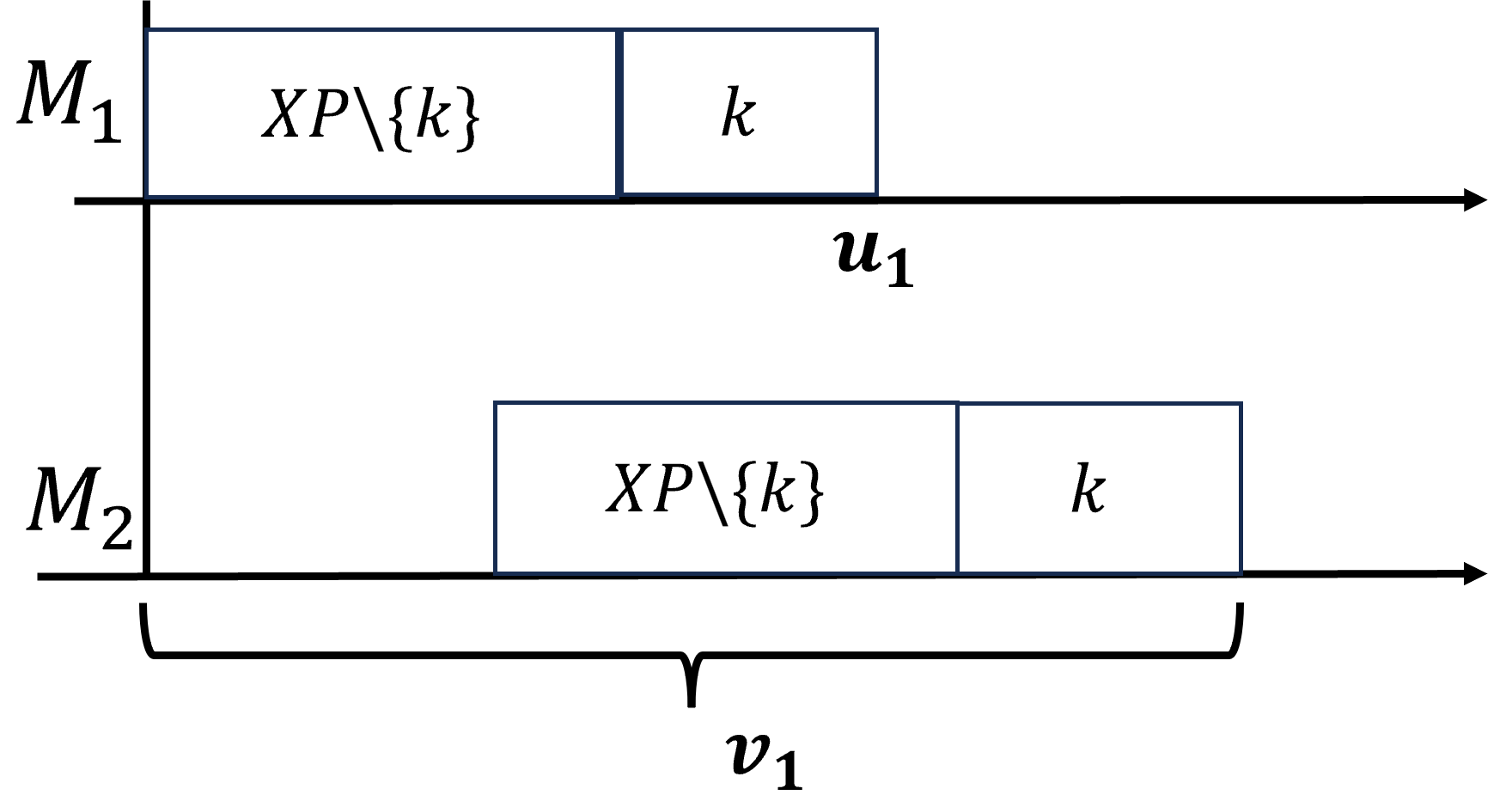}
			\caption{$u_1 + b_k < v_1$}
			\label{fig4:suba}
		\end{subfigure}
		\hfill
		\begin{subfigure}[b]{0.45\textwidth}
			\centering
			\includegraphics[width=\textwidth]{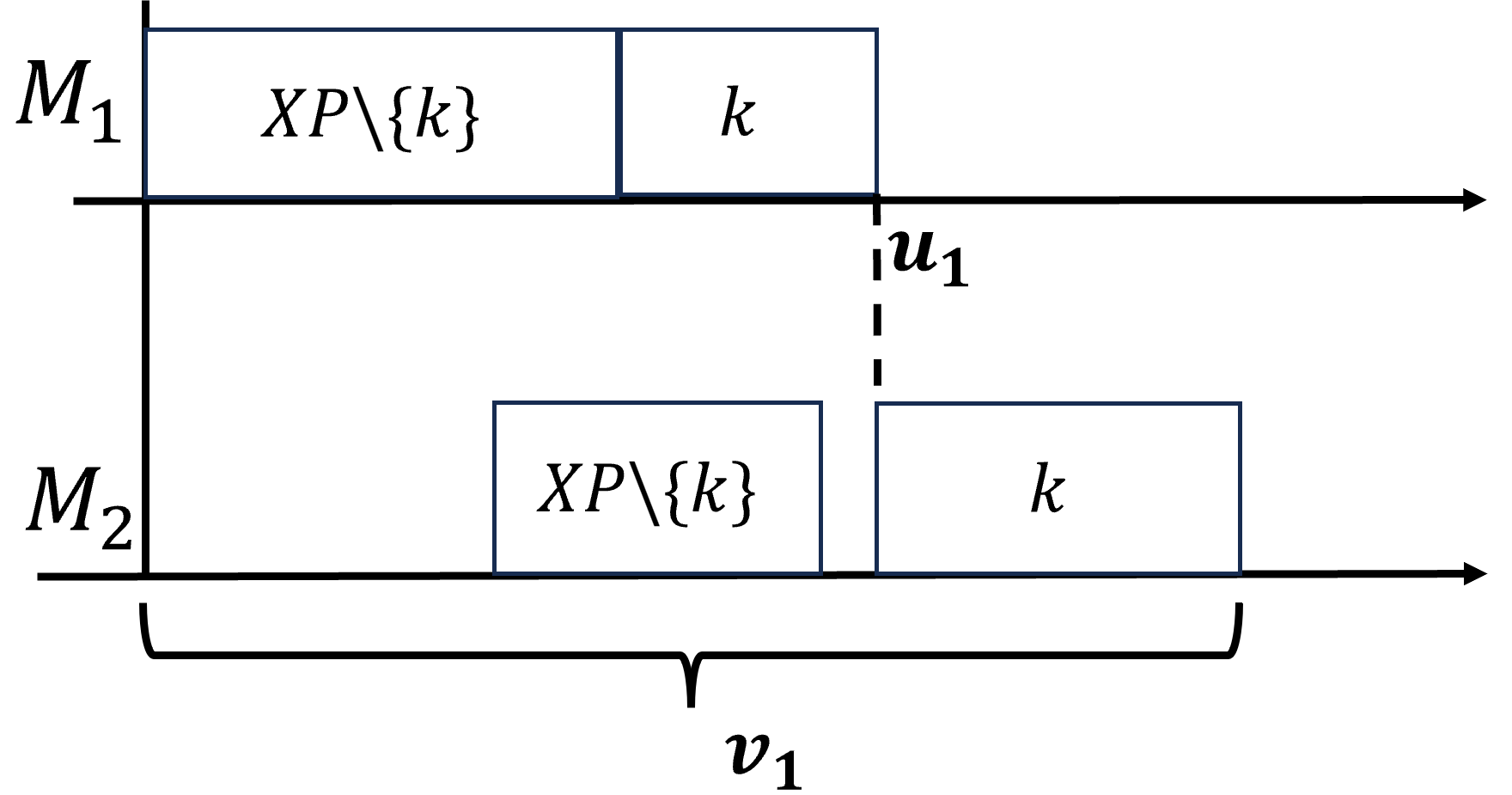}
			\caption{$u_1 + b_k = v_1$}
			\label{fig4:subb}
		\end{subfigure}
		\caption{Assigning job $k$ to $XP$.}
		\label{fig4}
	\end{figure}
	
	\paragraph{Part (b): assigning job \texorpdfstring{$k$ to $XS$}{k to XS}.}
	
	Let $v_2'$ be the value of $v_2$ before job $k$ is assigned. By the same argument as in part~(a), $v_2'=v_2-b_k$. Hence (Figure~\ref{fig5}),
	\[f_k^a = f_{k-1}^a (u_1, u_2-b_k, u_3, v_1, v_2-b_k, v_3)\]
	
	\begin{figure}[ht]
		\centering
		\begin{subfigure}[b]{0.45\textwidth}
			\centering
			\includegraphics[width=\textwidth]{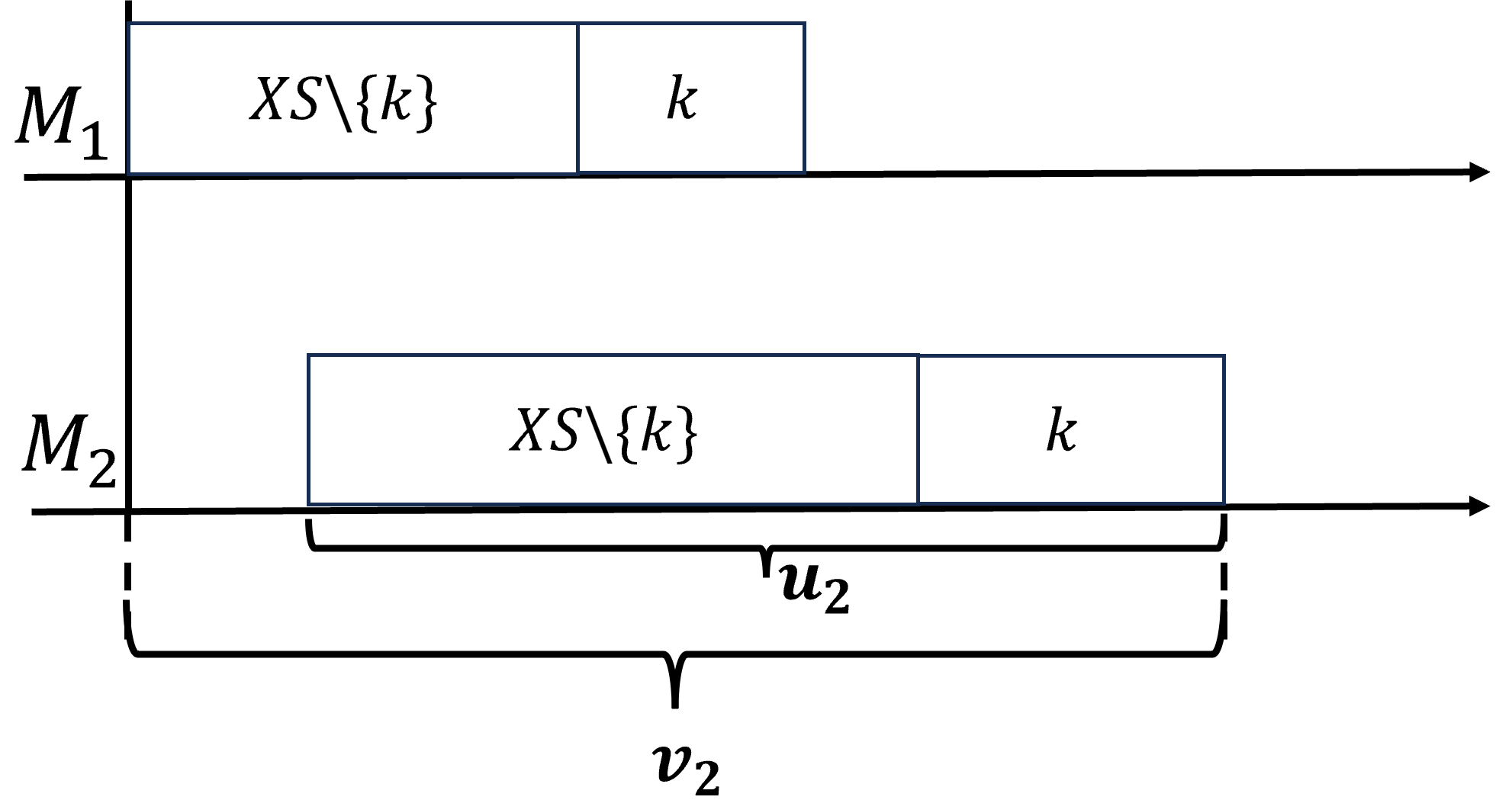}
			\caption{$a(XS) + b_k < v_2$}
			\label{fig5:suba}
		\end{subfigure}
		\hfill
		\begin{subfigure}[b]{0.45\textwidth}
			\centering
			\includegraphics[width=\textwidth]{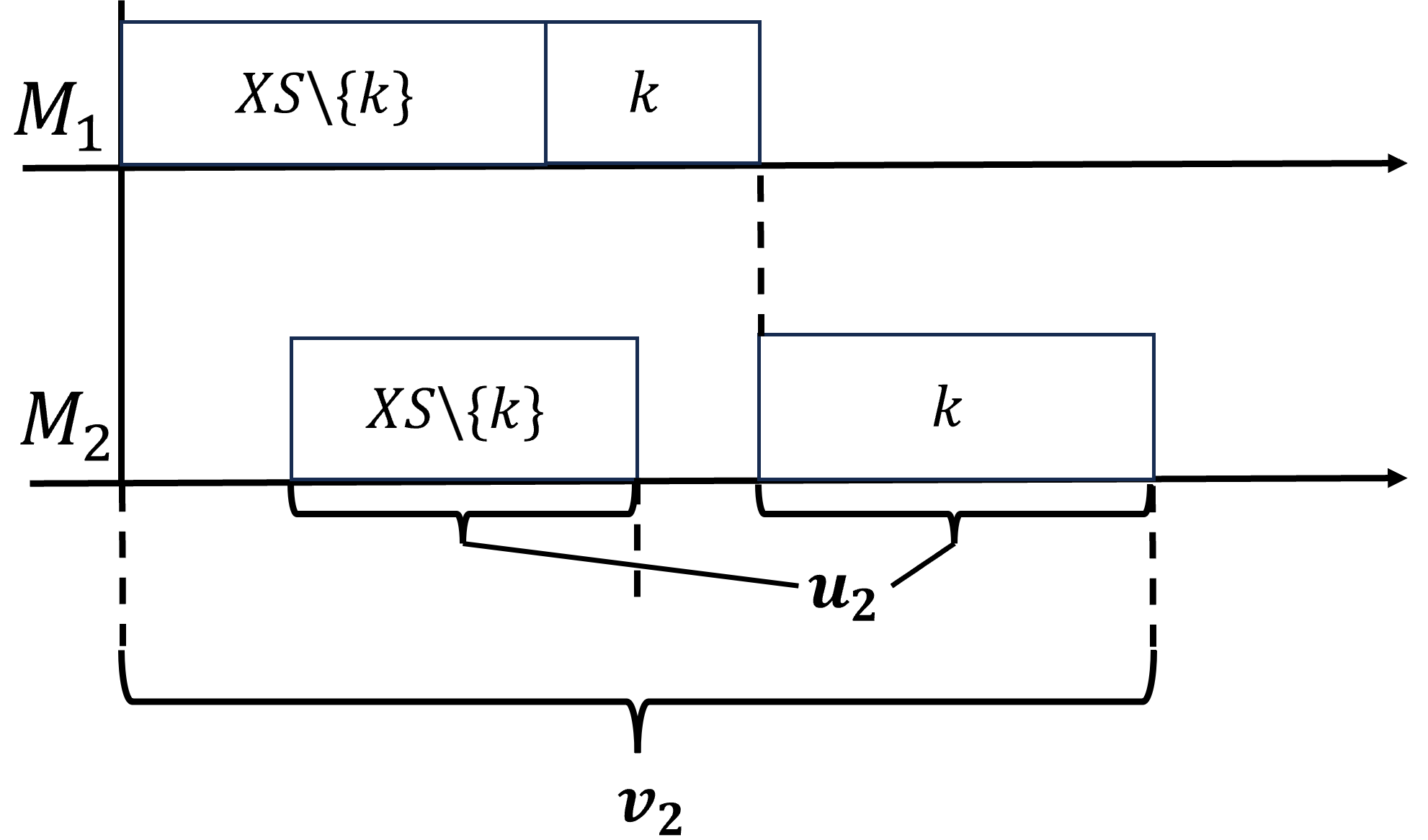}
			\caption{$a(XS) + b_k = v_2$}
			\label{fig5:subb}
		\end{subfigure}
		\caption{Assigning job $k$ to $XS$.}
		\label{fig5}
	\end{figure}
	
	\paragraph{Part (c): assigning job \texorpdfstring{$k$ to $YS$}{k to YS}.}
	If job $k$ is assigned to $YS$, it is the last job in that subset to begin processing. Let $v_3'$ be the value of $v_3$ before this assignment. As in part~(a), $v_3'=v_3-b_k$, and therefore (Figure~\ref{fig6})
	\[
	f_k^a = f_{k-1}^a (u_1, u_2, u_3-a_k, v_1, v_2, v_3-b_k)
	\]
	
	\begin{figure}[ht]
		\centering
		\begin{subfigure}[b]{0.45\textwidth}
			\centering
			\includegraphics[width=\textwidth]{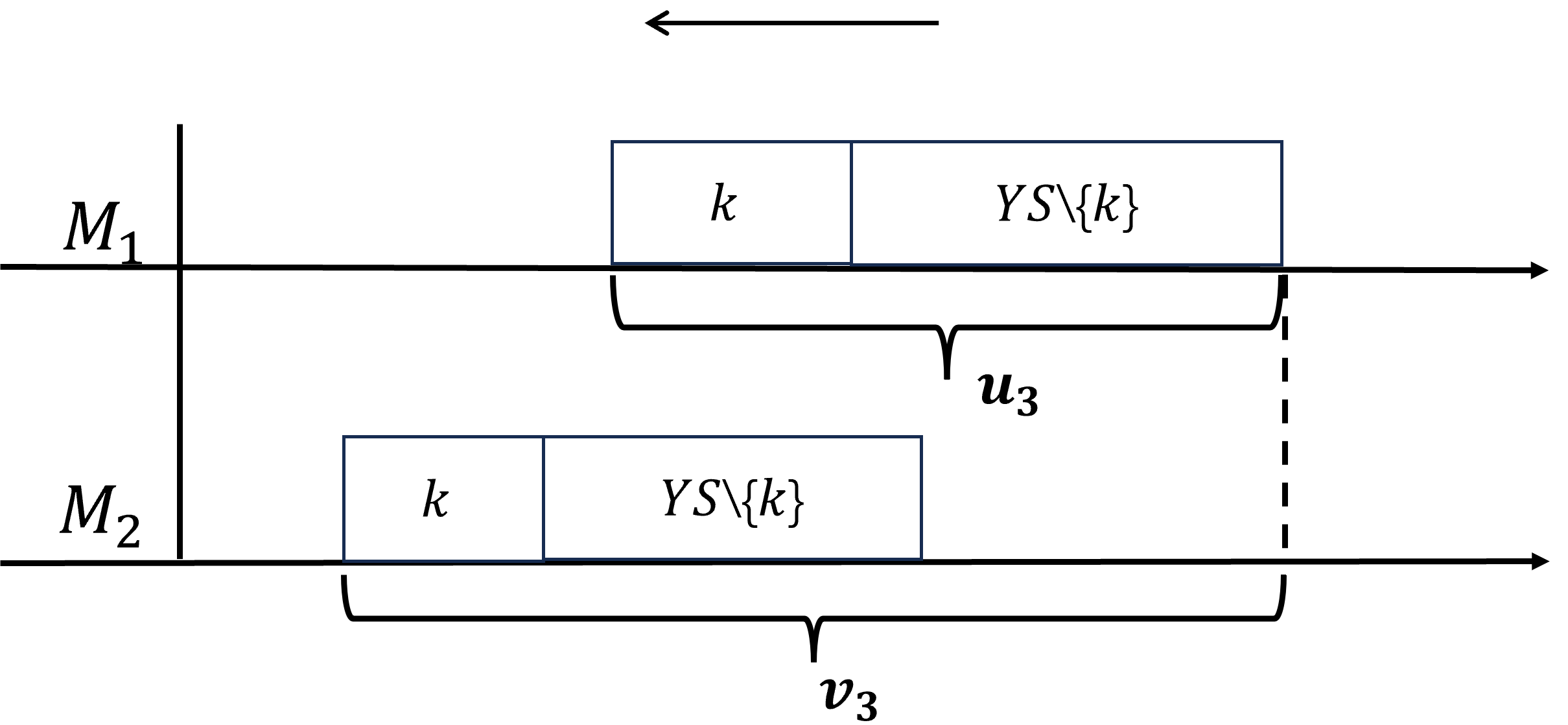}
			\caption{$u_3 + b_k < v_3$}
			\label{fig6:suba}
		\end{subfigure}
		\hfill
		\begin{subfigure}[b]{0.45\textwidth}
			\centering
			\includegraphics[width=\textwidth]{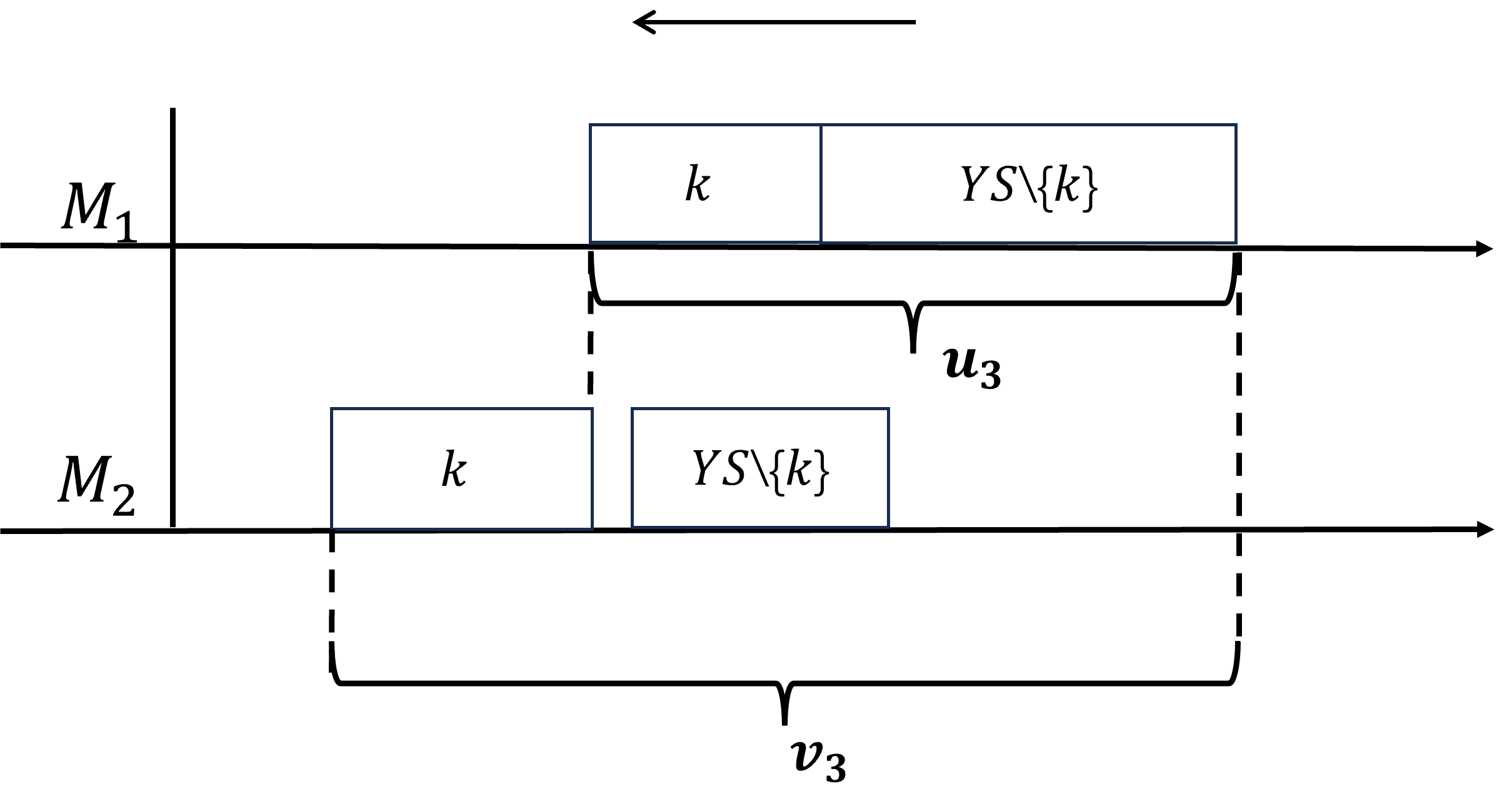}
			\caption{$u_3 + b_k = v_3$}
			\label{fig6:subb}
		\end{subfigure}
		\caption{Assigning job $k$ to $YS$.}
		\label{fig6}
	\end{figure}
	\subsubsection{Case B: \texorpdfstring{\(XS=\emptyset\)}{XS = empty}}
	Case B is handled analogously.
	
	Let $f_k^b(u_1,u_2,u_3,v_1,v_2,v_3)$ be the minimum value of $a(N_k\cap XP)$, evaluated using the original processing times, among all partial assignments represented by the state. We abbreviate this value to $f_k^b$ when the arguments are clear. The state variables are
	\begin{alignat*}{2}
		u_1 &= a(N_k \cap XP) & \; &\text{total processing time on } M_1 \text{ of } N_k \cap XP; \\
		u_2 &= b(N_k \cap YP) & \; &\text{total processing time on } M_2 \text{ of } N_k \cap YP; \\
		u_3 &= a(N_k \cap YS) & \; &\text{total processing time on } M_1 \text{ of } N_k \cap YS; \\
		v_1 &= C_{\max}(N_k \cap XP) & \; &\text{makespan of } N_k \cap XP; \\
		v_2 &= C_{\max}(N_k \cap YP) & \; &\text{makespan of } N_k \cap YP; \\
		v_3 &= C_{\max}(N_k \cap YS) & \; &\text{makespan of } N_k \cap YS.
	\end{alignat*}
	
	Let $h_k = a(N_k) - u_1 - u_3$.
	
	\noindent\textit{Initial conditions.}
	\[
	f_k^b =
	\begin{cases}
		0,
		& \textrm{if} \quad (k, u_1, u_2, u_3, v_1, v_2, v_3) = (0, 0, 0, 0, *, *, *) \\[4pt]
		+\infty,
		& \textrm{otherwise}
	\end{cases}
	\]
	Here $*$ denotes an arbitrary nonnegative value.
	
	\noindent\textit{Recurrence.}
	\[
	f_k^b =
	\min\begin{cases}
		f_{k-1}^b (u_1-a_k,u_2,u_3,v_1-b_k,v_2,v_3) + a_k,
		& \textrm{if} \quad u_1 + b_k \leq v_1, \quad (a) \\[4pt]
		f_{k-1}^b (u_1,u_2-b_k,u_3,v_1,v_2-b_k,v_3),
		& \textrm{if} \quad h_k + b_k \leq v_2, \quad (b) \\[4pt]
		f_{k-1}^b (u_1,u_2,u_3-a_k,v_1,v_2,v_3-b_k),
		& \textrm{if} \quad u_3 + b_k \leq v_3, \quad (c) \\[4pt]
		+\infty,
		& \textrm{otherwise.}
	\end{cases}
	\]
	The ranges are $k\in\{0,1,\ldots,n\}$, $u_1\in[0,s]$, $u_2\in[0,b(N_k)]$, $u_3\in[0,a(N_k)]$, $v_1\in[u_1,V]$, $v_2\in[u_2,V]$, and $v_3\in[u_3,V]$.
	
	Transitions (a), (b), and (c) assign job $k$ to $XP$, $YP$, and $YS$, respectively.
	
	\subsubsection{Makespan}
	
	The dynamic program produces the terminal states $f_n^a$ and $f_n^b$. Let $C_A$ and $C_B$ denote the minimum makespans in Cases A and B, respectively. By Proposition~\ref{prop:makespan-a},
	\[
	\begin{split}
	C_A = \min\limits_{(u_1, u_2, u_3, v_1, v_2, v_3)}\bigl\{
	&\max\{a(N)+d, v_3, e + u_3, v_1 + u_2,\\
	&\qquad u_1 + d + v_2, b(N)\}\mid f_n^a \leq s\bigr\},
	\end{split}
	\]
	where the minimum ranges over all terminal values of $u_1,u_2,u_3,v_1,v_2$, and $v_3$. Similarly, Proposition~\ref{prop:makespan-b} gives
	\[
	\begin{split}
	C_B = \min\limits_{(u_1, u_2, u_3, v_1, v_2, v_3)}\bigl\{
	&\max\{b(N), v_1, v_2 + u_3 + d, u_2 + v_3,\\
	&\qquad a(N) + d, e + u_3\}\mid f_n^b \leq s\bigr\}.
	\end{split}
	\]
	The optimal makespan for $O2\mid r\text{-}a(M_1)\mid C_{\max}$ is therefore
	\[
	\min\{C_A, C_B\}.
	\]
	
	The partition into $XP$, $XS$, $YP$, and $YS$ is recovered by backtracking from a terminal state attaining this minimum. The resulting partition, together with the prescribed Johnson orders, determines the corresponding optimal schedule.
	
	\section{FPTAS}\label{FPTAS}
	
	We now apply standard scaling and rounding to the dynamic program. For brevity, let $P$ denote an instance of $O2\mid r\text{-}a(M_1)\mid C_{\max}$, and let $\mathrm{OPT}$ be its optimal makespan. Define
	\[
	\mathrm{LB}=\max\{a(N),b(N),e,C_{\max}(\widehat{S}^*)\}.
	\]
	\begin{lemma}\label{lem:lower-bound}
		The lower bound $\mathrm{LB}$ satisfies $\mathrm{LB}\geq \mathrm{OPT}/2$.
	\end{lemma}
	\begin{proof}
		\[
		\mathrm{OPT}
		\leq e+C_{\max}(\widehat{S}^*)
		\leq 2\max\{e,C_{\max}(\widehat{S}^*)\}
		\leq 2\mathrm{LB},
		\]
		and consequently $\mathrm{LB}\geq\mathrm{OPT}/2$.
	\end{proof}
	Let $\mathrm{UB}$ be the makespan returned by the algorithm of Breit et al.~\cite{breit2001two}; then $\mathrm{UB}\leq 4\mathrm{OPT}/3$.
	
	Fix $\varepsilon>0$ and set
	\[
	\delta=\frac{\varepsilon\,\mathrm{LB}}{n+2}.
	\]
	
	The rounded instance $P^0$ is obtained by defining
	\begin{equation}
		d^0 = \lfloor d/\delta \rfloor, \quad a_i^0 = \lfloor a_i/\delta \rfloor, \quad b_i^0 = \lfloor b_i/\delta \rfloor, \quad s^0 = \lfloor s/\delta \rfloor \quad (i = 1,2,\ldots,n). \label{equation2}
	\end{equation}
	Set $e^0=s^0+d^0$. For each rounded state, define $h_k^0=a^0(N_k)-u_1^0-u_3^0$.
	
	We bound every rounded state coordinate by
	\begin{align*}
		\frac{\mathrm{UB}}{\delta}
		\leq \frac{8\mathrm{LB}}{3\delta}
		\leq \left\lceil\frac{8(n+2)}{3\varepsilon}\right\rceil
		=:V^0,
	\end{align*}
	where the first inequality follows from $\mathrm{UB}\leq 8\mathrm{LB}/3$.
	
	Replacing the state coordinates by their rounded counterparts yields the following recurrence for Case A; Case B is modified analogously:
	\[
	f_k^a =
	\min\begin{cases}
		f_{k-1}^a (u_1^0 - a_k^0, u_2^0, u_3^0, v_1^0 - b_k^0, v_2^0, v_3^0) + a_k,
		& \textrm{if} \quad u_1^0 + b_k^0 \leq v_1^0 \quad  \\[4pt]
		f_{k-1}^a (u_1^0, u_2^0 - b_k^0, u_3^0, v_1^0, v_2^0 - b_k^0, v_3^0),
		& \textrm{if} \quad h_k^0 + b_k^0 \leq v_2^0 \quad  \\[4pt]
		f_{k-1}^a (u_1^0, u_2^0, u_3^0 - a_k^0, v_1^0, v_2^0, v_3^0 - b_k^0),
		& \textrm{if} \quad u_3^0 + b_k^0 \leq v_3^0 \quad  \\[4pt]
		+\infty,
		& \textrm{otherwise}
	\end{cases}
	\]
	Here $k\in\{0,1,\ldots,n\}$, $u_1^0\in[0,s^0]$, $u_2^0\in[0,b^0(N_k)]$, $u_3^0\in[0,a^0(N_k)]$, $v_1^0\in[u_1^0,V^0]$, $v_2^0\in[u_2^0,V^0]$, and $v_3^0\in[u_3^0,V^0]$.
	
	Although the state coordinates are rounded, transition~(a) adds the original processing time $a_k$, rather than $a_k^0$, to the value function. Consequently, the condition $f_n\leq s$ remains valid when the rounded schedule is mapped back to the original instance.
	\begin{mdframed}[
		linewidth=1pt,
		leftmargin=10pt,
		rightmargin=10pt
		]
		\textbf{Algorithm FP}
		
		\textbf{Step 0.} Given an instance $P$ and $\varepsilon>0$, index the jobs according to Johnson's rule.
		
		\textbf{Step 1.} Construct the rounded instance $P^0$ using~\eqref{equation2}.
		
		\textbf{Step 2.} Run the rounded dynamic program on $P^0$.
		
		\textbf{Step 3.} Use Propositions~\ref{prop:makespan-a} and~\ref{prop:makespan-b} to evaluate every terminal state ($k=n$). Choose a state of minimum makespan and recover the corresponding schedule $\pi^0$ by backtracking. Denote its makespan in $P^0$ by $C^0_{\max}(\pi^0)$.
		
		\textbf{Step 4.} Apply the order specified by $\pi^0$ to the original instance $P$, scheduling every operation as early as possible. Return the resulting schedule, whose makespan is denoted by $C_{\max}(\pi^0)$.
		
	\end{mdframed}
	
	\begin{theorem}
		Algorithm FP is an FPTAS for $O2\mid r\text{-}a(M_1)\mid C_{\max}$.
	\end{theorem}
	
	\begin{proof}
		Let $P'$ be the instance obtained from $P$ by scaling every time parameter by $1/\delta$; that is,
		
		\begin{equation}
			d'=d/\delta,\quad s'=s/\delta,\quad
			a_i'=a_i/\delta,\quad b_i'=b_i/\delta
			\quad (i=1,2,\ldots,n).\label{equation3}
		\end{equation}
		
		Let $\pi^*$ be an optimal schedule for $P$.
		
		Let $C_{\max}'(\pi^0)$ denote the makespan obtained by applying $\pi^0$ to $P'$. Exact scaling gives
		\begin{equation}
			C_{\max}(\pi^0)=\delta C_{\max}'(\pi^0). \label{eq4}
		\end{equation}
		
		Each rounded time parameter differs from its scaled counterpart by less than one. Since a critical completion-time expression contains at most $n+1$ operation lengths and one unavailability interval,
		\begin{equation}
			C_{\max}'(\pi^0)\leq C_{\max}^0(\pi^0)+n+2. \label{eq5}
		\end{equation}
		
		Let $C_{\max}^0(\pi^*)$ be the makespan obtained by applying $\pi^*$ to $P^0$. Since $\pi^0$ is optimal for $P^0$,
		\begin{equation}
			C_{\max}^0(\pi^0)\leq C_{\max}^0(\pi^*). \label{eq6}
		\end{equation}
		
		Finally, rounding down cannot increase the makespan, so
		\begin{equation}
			C_{\max}^0(\pi^*)\leq C_{\max}(\pi^*)/\delta. \label{eq7}
		\end{equation}
		
		Combining~\eqref{eq4}--\eqref{eq7} yields
		\begin{align*}
			C_{\max}(\pi^0)
			&=\delta C_{\max}'(\pi^0) && \text{by~\eqref{eq4}}\\
			&\leq\delta\bigl(C_{\max}^0(\pi^0)+n+2\bigr) && \text{by~\eqref{eq5}}\\
			&\leq\delta\bigl(C_{\max}^0(\pi^*)+n+2\bigr) && \text{by~\eqref{eq6}}\\
			&\leq C_{\max}(\pi^*)+(n+2)\delta && \text{by~\eqref{eq7}}\\
			&=C_{\max}(\pi^*)+\varepsilon\,\mathrm{LB}\\
			&\leq(1+\varepsilon)C_{\max}(\pi^*).
		\end{align*}
		
		The rounded dynamic program has $O(n(V^0)^6)$ states. Since $V^0=\lceil 8(n+2)/(3\varepsilon)\rceil$, its running time is $O(n^7/\varepsilon^6)$, which is polynomial in both $n$ and $1/\varepsilon$.
	\end{proof}
	
	\section{Conclusion}\label{Conclusion}
	We studied two-machine open-shop scheduling with a single fixed MP in the resumable setting. A four-way partition of the jobs allows the makespan to be recovered from a compact collection of flow-shop states. This structure yields a seven-dimensional pseudo-polynomial dynamic program, improving on the previous ten-dimensional formulation, and leads to the first FPTAS for $O2\mid r\text{-}a(M_1)\mid C_{\max}$. An interesting direction for future work is to determine whether the corresponding non-resumable problem also admits an FPTAS.

\backmatter

\bmhead{Acknowledgements}

This work was supported by the National Natural Science Foundation of China
(Grant Nos. 12471338, 12461011, and 12371363) and the Liaoning Revitalization
Talents Program (Grant No. XLYC2403125).


\begingroup
\setlength{\bibsep}{0.8em}
\bibliography{algorithmica-article}
\endgroup

\end{document}